\newif\ifdraft\draftfalse  
\newif\ifanon\anonfalse     
\newif\iffull\fullfalse   
\newif\ifbackref\backreffalse 
\newif\ifsooner\soonerfalse
\newif\iflater\laterfalse
\newif\ifieee\ieeefalse
\makeatletter \@input{texdirectives} \makeatother

\documentclass[10pt, conference, compsocconf, letterpaper, times]{IEEEtran}

\def\IEEEbibitemsep{3pt}

\usepackage{etex}
\usepackage{amsmath}
\usepackage{mathtools}
\usepackage{url}
\usepackage{xspace}
\usepackage{multirow}
\usepackage{array}
\usepackage{amsmath,amssymb,amsthm}
\usepackage{stmaryrd}
\SetSymbolFont{stmry}{bold}{U}{stmry}{m}{n} 
\usepackage{xcolor}
\usepackage{listings}
\usepackage{paralist}
\usepackage{fancyvrb}
\usepackage{latexsym}
\usepackage{bm}
\usepackage{morefloats}
\usepackage[shortlabels,inline]{enumitem}
\usepackage{fancyhdr}
\usepackage[yyyymmdd,hhmmss]{datetime}
\usepackage{graphicx}
\usepackage[noadjust]{cite}
\usepackage{coqed}

\usepackage{tikz}
\usetikzlibrary{matrix,arrows,decorations.pathmorphing}

\usepackage{preamble}
\typicallabel{MkKey}
\renewcommand{\andalso}{\quad\;\;}

\begin{document}

\title{{\Huge\bf Beyond Good and Evil}
\\[1ex]
\Large\bf Formalizing the Security Guarantees of Compartmentalizing Compilation}

\author{
\ifanon
\else
  Yannis Juglaret\textsuperscript{1,2} \quad 
  C\u{a}t\u{a}lin Hri\c{t}cu\textsuperscript{1} \quad 
  Arthur Azevedo de Amorim\textsuperscript{4} \quad 
  Boris Eng\textsuperscript{1,3} \quad 
  Benjamin C. Pierce\textsuperscript{4} \\[1ex] 
  \textsuperscript{1}Inria Paris\qquad
  \textsuperscript{2}Universit\'{e} Paris Diderot (Paris 7)\qquad
  \textsuperscript{3}Universit\'{e} Paris 8\qquad
  \textsuperscript{4}University of Pennsylvania
\fi
}

\maketitle

\begin{abstract}
Compartmentalization is good security-engineering practice. By
breaking a large software system into mutually distrustful components
that run with minimal privileges, restricting their interactions to
conform to well-defined interfaces, we can limit the damage caused by
low-level attacks such as control-flow hijacking.  When used to defend
against such attacks, compartmentalization is often implemented
cooperatively by a compiler and a low-level compartmentalization
mechanism. However, the formal guarantees provided by such {\em
  compartmentalizing compilation} have seen surprisingly little
investigation.

We propose a new security property, \emph{secure compartmentalizing
  compilation (SCC)}, that formally characterizes the guarantees
provided by compartmentalizing compilation and clarifies its attacker
model. We reconstruct our property by starting from the
well-established notion of fully abstract compilation, then
identifying and lifting three important limitations that make standard
full abstraction unsuitable for compartmentalization. The connection
to full abstraction allows us to prove SCC by adapting established
proof techniques; we illustrate this with a compiler from a simple
unsafe imperative language with procedures to a compartmentalized
abstract machine.
\end{abstract}




\section{Introduction}
\label{sec:intro}

Computer systems are distressingly insecure.
Visiting a website, opening an email, or serving a client request is often
all it takes to be subjected to a control-hijacking attack.
%
%
These devastating low-level attacks typically exploit memory-safety
vulnerabilities such as buffer overflows, use-after-frees, or double
frees, which are abundant in large software systems.
Various techniques have been
proposed for guaranteeing memory safety~\cite{NagarakatteZMZ09, NagarakatteZMZ10, DeviettiBMZ08,
  Nagarakatte2013, NagarakatteMZ14, NagarakatteMZ15,
  interlocks_ahns2012, micropolicies2015, LowFat2013}, but
the challenges of efficiency~\cite{NagarakatteZMZ09, NagarakatteZMZ10},
precision~\cite{m7}, scalability~\cite{ZitserLL04}, backwards
compatibility~\cite{cheri_asplos2015}, and effective
deployment~\cite{DeviettiBMZ08, Nagarakatte2013, NagarakatteMZ14,
  NagarakatteMZ15, interlocks_ahns2012, micropolicies2015, LowFat2013,
  pump_asplos2015}
have hampered their widespread adoption.

Meanwhile, new mitigation techniques have been proposed to deal with the most
onerous consequences of memory  unsafety---for instance, techniques
aimed at preventing
control-flow hijacking even in unsafe
settings~\cite{Abadi2005, AbadiBEL09, Erlingsson07, TiceRCCELP14, BurowCBPNLF16}.
Unfortunately, these defenses often underestimate the power of the attackers
they may face~\cite{Erlingsson07, SnowMDDLS13, outofcontrol_ieeesp2014,
  DaviSLM14, EvansLOSROS15, EvansFGOTSSRO15}---if, indeed, they have any
clear model at all of what they are protecting against.
Clarifying the precise security properties and
attacker models of practical mitigation techniques is thus an important
research problem---and a challenging one, since a good model has to capture
not only the defense mechanism itself but also the essential features of the
complex world in which low-level attacks occur.


In this paper we focus on the use of
{\em compartmentalization}~\cite{GudkaWACDLMNR15, cheri_oakland2015,
  wedge_nsdi2008} as a strong, practical defense mechanism
against low-level attacks exploiting memory unsafety.
The key idea is to break up a large software system into
mutually distrustful components that run with minimal
privileges and can interact only via well-defined interfaces.
This is not only good software engineering; it also gives strong security
benefits.  In particular, control-hijacking attacks can compromise only
specific components with exploitable vulnerabilities, and thus only give the
attacker direct control over the privileges held by these components.
Also, because compartmentalization can be enforced by more coarse-grained
mechanisms, acceptable efficiency and backwards compatibility are generally
easier to achieve than for techniques enforcing full-blown memory safety.
\ch{They also protect against potentially malicious code. Although
  that's not the main focus here, things like NaCl are made to isolate
  untrusted plugins.}

When used as a defense mechanism against memory unsafety,
compartmentalization is often achieved via cooperation
between a compiler and a low-level compartmentalization
mechanism~\cite{KrollSA14, ZengTE13, JuglaretHAPST15, GudkaWACDLMNR15,
  PatrignaniDP16, cheri_oakland2015, cheri_asplos2015}.
In this paper we use {\em compartmentalizing compilation} to refer to
cooperative implementations of this sort.
%
The compiler might, for instance, insert dynamic checks and cleanup
code when switching components and provide information about
components and their interfaces to the low-level compartmentalizing
mechanism, which generally provides at least basic
isolation.
%
Two such low-level compartmentalization technologies are already widely
deployed: process-level privilege separation~\cite{Kilpatrick03,
  GudkaWACDLMNR15, wedge_nsdi2008} (used, \EG by OpenSSH~\cite{ProvosFH03}
and for sandboxing plugins and tabs in modern web browsers~\cite{ReisG09})
and software fault isolation~\cite{sfi_sosp1993} (provided, \EG by
Google Native Client~\cite{YeeSDCMOONF10}); many more
are on the drawing boards~\cite{micropolicies2015, sgx, PatrignaniDP16,
  cheri_oakland2015, cheri_asplos2015}.


So what security guarantees does compartmentalizing compilation provide,
and what, exactly, is its attacker model?
%
%
\ch{The usual? No longer think it's good but it's standard.}%
A good starting point for addressing these questions is the familiar notion
of {\em fully abstract compilation}~\cite{abadi_protection98,
  PatrignaniASJCP15, AgtenSJP12, abadi_aslr12, JagadeesanPRR11,
  FournetSCDSL13, AbadiPP13, AbadiP13, AhmedB11, AhmedB08, NewBA16}.
A fully abstract compiler toolchain (compiler, linker, loader, and
underlying architecture with its security mechanisms) protects the
interactions between a compiled program and its low-level environment,
allowing programmers to reason
soundly about the behavior of their code when it is placed in an arbitrary
target-language context, by considering only its behavior in arbitrary
source-language contexts.
In particular, if we link the code produced by such a compiler against
arbitrary low-level libraries---perhaps compiled from an unsafe language or
even written directly in assembly---%
the resulting execution will not be any less
secure than if we had restricted ourselves to library code written in the
same high-level language as the calling program.  

(Why is it useful to
restrict attention to attackers written in a high-level language?  First,
because reasoning about what attackers might do---in particular, what
privileges they might exercise---is easier in a high-level language.  And
second, because by phrasing the property in terms of low- and
high-level programs rather than directly in terms of attacker behaviors,
specific notions of privilege, etc., we can re-use the same property for
many different languages.)

Since full abstraction works by partitioning the world into a program and
its context, one might
expect it to apply to compartmentalized programs as well:
some set of components that are assumed to be subject to control-hijacking
attacks could be
grouped into the ``low-level context,'' while some others that are assumed
to be immune to 
such attacks would constitute the ``high-level program.''  Full
abstraction would then allow us to reason about the possible behaviors
of the whole system using the simplifying assumption that the
attacker's injected behavior for the compromised components can be
expressed in the same high-level language as the good components.
Sadly, this intuition does not withstand closer examination.  Full
abstraction, as previously formulated in the literature, suffers
from three important limitations that make it unsuitable for characterizing
the security guarantees of compartmentalizing compilation.

First, fully abstract compilation assumes that the source language itself is
secure,
so that it makes sense to define target-level security with respect to the
semantics of the source language.
However, compartmentalization is often applied to languages like C and C++,
which do {\em not} have a secure semantics---the C and C++ standards leave
most of the security burden to the programmer by calling out a large number
of {\em undefined behaviors}, including memory-safety violations, that are
assumed never to occur.
Valid compilers for these languages are allowed to generate code that does
literally {\em anything}---in particular, anything a remote
attacker may want---when applied to inputs that lead to undefined behavior.
There is no way to tell, statically, whether or not a program may have
undefined behavior, and compilers do not check for this situation.  (Indeed,
not only do they not check: they aggressively exploit the assumption of no
undefined behaviors to produce the fastest possible code for well-defined
programs, often leading to easily exploitable behaviors when this assumption
is broken.)
The point of compartmentalizing compilation
is to ensure that the potential effects of undefined
behavior are limited to the compromise of the component in which it occurs:
other components can only be influenced by compromised
ones via controlled interactions respecting specified interfaces.

To characterize the security of compartmentalizing
compilation, we therefore need a formal property that can meaningfully
accommodate source languages in which components can be compromised
via undefined behavior.
Full abstraction as conventionally formulated does not fit the bill,
because, in order to preserve equivalences of programs with undefined
behavior, compilers must abandon the aggressive optimizations that are
the reason for allowing undefined behaviors in the first place.
To see this, consider C expressions {\tt buf[42]} and {\tt
  buf[43]} that read at different positions {\em outside} the bounds of a
buffer {\tt buf}.
These two programs are equivalent at the source level: they both
lead to arbitrary behavior.
However, a real C compiler would never compile these expressions to
equivalent code, since this would require runtime checks that many C
programmers would deem too expensive.

Second, fully abstract compilation makes an {\em open world} assumption
about the attacker context. While the context is normally required to
be compatible with the protected program, for instance by respecting
the program's typed interface, the structure and privilege of the context are
unrestricted (the full abstraction definition quantifies over {\em
  arbitrary} low-level contexts).
%
This comes in direct contradiction with the idea of least
privilege%
,
which is crucial to compartmentalization, and which relies on the fact
that even if a component is compromised, it does not immediately get
more privilege.
Compromised components cannot change the basic rules of the
compartmentalization game.
For instance, in this paper we consider a static compartmentalization
setting, in which the breakup of the application into components is
fixed in advance, as are the privileges of each component.
A security property suitable for this setting needs to be restricted
to contexts that conform to a fixed breakup into components with
static privileges.%
\footnote{%
  In a setting where new components can be dynamically created and
  privileges can be exchanged dynamically between components, the
  details of this story will be more complicated; still, we expect any
  secure compartmentalizing compilation property to limit the ability
  of low-level attacker contexts to ``forge'' the privileges of
  existing components.}
%

Third, because the definition of full abstraction involves applying the
compiler only to a program and not to the untrusted context in which it
runs, a fully abstract compiler may choose to achieve its protection goals
by introducing just a single barrier around the trusted part to protect it
from the untrusted part~\cite{PatrignaniASJCP15, AgtenSJP12, LarmuseauPC15,
  PatrignaniCP13, patrignani_thesis}.
Such compilation schemes force the programmer to commit in advance to a
single compromise scenario, \IE to a single static split of their
application into a ``good'' trusted program and an ``evil'' untrusted
context from which this program has to be protected.
This is not realistic in the setting of compartmentalizing compilation,
where we generally cannot predict which components may be vulnerable to
compromise by control hijacking attacks, and instead must simultaneously
guard against multiple compromise scenarios.
Compartmentalizing compilers allow us to build more secure applications
that go beyond the blunt trusted/untrusted distinction made by some fully
abstract compilers.
To describe their guarantees accurately, we thus need a new property that
captures the protection obtained by breaking up applications into multiple
mutually distrustful components, each running with least privilege, and that
permits reasoning about multiple scenarios in which different subsets of
these components are compromised.
%

Our main contribution is the definition of such a property,
which we call {\em secure compartmentalizing compilation (SCC)}
(\autoref{sec:sc}).
While similar in many respects to full abstraction, our property
escapes the three limitations discussed above.
First, it
applies to unsafe source languages with
undefined behaviors by introducing a new notion of {\em fully defined}
sets of components.
While undefined behavior is a property of whole programs, full definedness
is compositional.
Intuitively, a set of components is fully defined if they cannot be
{\em blamed}~\cite{FindlerF02prime} for undefined behavior
in any context satisfying fixed interfaces.
Second, SCC makes a {\em closed-world} assumption about compromised
components, enforcing the basic rules of the compartmentalization game
like the fixed division into components and the fixed privileges of
each component, including for instance with which other components it
is allowed to interact.
%
%
Third, SCC ensures protection for
multiple, mutually distrustful components; it does not assume we know in
advance which components are going to be compromised (i.e., in the C
setting, which components may contain exploitable undefined behaviors), but
instead explicitly quantifies over all possible compromise scenarios.

Our second contribution is relating SCC
to standard formulations of full abstraction both intuitively and formally
(\autoref{sec:fa-not-enough}).
We start from full abstraction and show how the three
limitations that make it unsuitable in our setting can be lifted
one by one.
This results in two properties we call {\em structured full abstraction} and
{\em separate compilation}, which can be combined and instantiated to obtain
SCC.
While our property directly captures the intuition of our
attacker model, reducing it to structured full abstraction is a useful
technical step, since the latter is easier to establish for specific
examples using a variant of existing proof techniques.
%
Moreover, arriving at the same property by two different paths
increases our confidence that we found the right property.

Our third contribution is establishing the SCC property for a simple
unsafe imperative language with components interacting via procedure
calls and returns, compiling to an abstract machine with protected
compartments (\autoref{sec:instance}).
Despite the simplicity of the setting, this result gives useful
insights.  
First, the source language and compilation strategy enable interesting
attacks on components with potential buffer overflows, similar to those
found in C.
Second, we illustrate how SCC can be
achieved by the cooperation of a compiler (cleaning and restoring
registers) and a low-level protection mechanism (totally isolating
compartments and providing a secure interaction mechanism using calls
and returns).
Third, our SCC proof adapts a standard technique called {\em trace
  semantics}~\cite{JeffreyR05, PatrignaniC15}, via the reduction to
structured full abstraction.
The closed-world assumption about the context made by structured full
abstraction requires some nontrivial changes to the trace
semantics proof technique.

The remainder of the paper describes each of our three contributions in
detail (\autoref{sec:sc}--\autoref{sec:instance}) and closes by discussing
related work (\autoref{sec:related}) and future directions
(\autoref{sec:conclusion}).
The supplemental materials \ifanon submitted \else
associated \fi with this paper includes:
(a) a Coq proof for \autoref{thm:sfa-to-sc};
(b)~technical details and proofs for the SCC instance from
    \autoref{sec:instance} (while most of these proofs are done only on
    paper, the main structured full abstraction result,
    \autoref{thm:sfa-instance}, is also proved in Coq); and
(c) a trace mapping algorithm in OCaml using property-based
    testing to support \autoref{assumption:definability}.
These materials can be found at:
\url{https://github.com/secure-compilation/beyond-good-and-evil}

\section{Secure Compartmentalizing Compilation}
\label{sec:sc}

We start with an intuitive explanation of compartmentalizing compilation,
its attacker model, and its security benefits, and then introduce {\em
  secure compartmentalizing compilation (SCC)}.

We consider compartmentalization mechanisms provided by the
compiler and runtime system for an unsafe programming language with some
notion of components.%
\footnote{We use the term ``runtime system'' loosely to include operating
  system mechanisms~\cite{Kilpatrick03, GudkaWACDLMNR15, wedge_nsdi2008,
    ProvosFH03, ReisG09} and/or hardware
  protections~\cite{micropolicies2015,sgx,PatrignaniDP16,cheri_oakland2015}
  that may be used by the compiler.}
In \autoref{sec:instance} we will present a simple example in detail,
but for the present discussion it suffices to think informally of C or
C++ enriched with some compartmentalization mechanism.
%
%
%
This mechanism allows security-conscious developers to break large
applications into mutually distrustful components running with
least privilege and interacting only via well-defined interfaces.
We assume that the interface of each component also gives a precise
description of its privilege.
Our notion of interface here is quite generic: interfaces might include any
information that can be dynamically enforced on components, including module
signatures, lists of allowed system calls, or more detailed access control
specifications describing legal parameters to inter-component calls (\EG
ACLs for files).
We assume that the division of the application into components and the
interfaces of those components are statically determined and fixed
throughout execution.
In \autoref{sec:instance}, we instantiate this picture with a rather simple
and rigid notion of components and interfaces, where components don't
directly share any state and where the only thing one component can do to
another one is to call the procedures allowed by the interfaces of both
components.

We do not fix a specific compartmentalizing compilation mechanism; we
just assume that whatever mechanism is chosen can guarantee that, even
if one component is compromised (e.g., by a control-hijacking attack),
it will still be forced to adhere to its specified interface in its
interactions with other components.
%
%
What a compromised component {\em can} do in this model is use its access
to other components, as allowed by its interface, to trick them into
misusing their own privileges (confused deputy attacks) and/or attempt to
mount further control-hijacking attacks on other
components by communicating with them via defined
interfaces.

We do not assume we know in advance which components will be
compromised: the compartmentalizing compilation mechanism has to
protect each component from all the others.
This allows developers to reason informally about various compromise
scenarios and their impact on the security of the whole
application~\cite{GudkaWACDLMNR15}, relying on conditional
reasoning of the form: ``If {\em these} components get taken over and {\em
  these} do not, then {\em this} might happen (while {\em that} cannot),
whereas if these other components get taken over, then this other thing
might happen...''
If the practical consequences of some plausible compromise scenario are too
serious,
developers can further reduce or separate privilege by
narrowing interfaces or splitting components, 
or they can make components more defensive by dynamically validating the
inputs they receive from other components.

For instance, developers of a compartmentalized web browser~\cite{ReisG09}
might reason about situations in which some subset of plugins and tabs gets
compromised and how this might impact the browser kernel and the remaining
plugins and tabs.
%
%
A possible outcome of this exercise might be noticing that, if the browser
kernel itself is compromised, then all bets are off for all the components
and the application as a whole, so the developers should put extra energy in
defending the kernel against attacks from compromised plugins or tabs.  On
the other hand, if interfaces {\em between} tabs and plugins are
appropriately limited, then compromise of one should not disrupt the rest.

Our goal is to articulate a security property that supports reasoning about
multiple compromise scenarios and clarifies the associated attacker model.
At the same time, our property is intended to serve as a benchmark for
developers of compartmentalizing compilation mechanisms who want to argue
formally that their mechanisms are secure.
In the rest of this section we explain the technical ideas behind the
SCC property and then give its formal definition.

An {\em application} is a set $\ii{Cs}$ of {\em components}, with
corresponding {\em interfaces} $\ii{CIs}$.
These components are separately compiled (individually compiling each
component in the set \ii{Cs} is written $\comp{\ii{Cs}}$) and linked
together (written $\link{\comp{\ii{Cs}}}$) to form an executable
binary for the application.
%

SCC quantifies over all {\em compromise
  scenarios}---i.e., over all ways of partitioning the components
into a set of compromised ones and a set of uncompromised ones.
In order to ensure that the set
  of compromised components doesn't expand during evaluation,
we require that the uncompromised components be {\em fully defined}
with respect to the interfaces of the compromised components.
That is, the uncompromised components must not exhibit undefined
behaviors
even if we replace the compromised components with arbitrary code (obeying
the same interfaces).


The full definedness condition is a necessary part of
the {\em static compromise model} considered in this paper.
Intuitively, if an uncompromised component can be tricked into an
undefined behavior by interface-respecting communication with other
components, then we need to conservatively assume that the already
compromised components will succeed in compromising this component
dynamically, so it belongs in the set of compromised components from
the start.
This static model is much simpler to reason about than a model of
dynamic compromise, in which one could perhaps provide guarantees to
not-fully-defined components up to the point at which they exhibit
undefined behavior, but which could, however, invalidate standard
compiler optimizations that involve code motion.
Moreover, it seems highly nontrivial to define our property for such a
more complex model.

\autoref{fig:compromise-scenario} illustrates one way to partition
five components $C_1,\dots,C_5$ with interfaces $i_1,\ldots,i_5$,
representing the scenario where $C_2$, $C_4$, and $C_5$ are
compromised and $C_1$ and $C_3$ are not.
In order for this compromise scenario to be considered by our property,
$C_1$ and $C_3$ need to be fully defined with respect to interfaces
$i_2$, $i_4$, and $i_5$, which means $C_1$ and $C_3$ cannot cause
undefined behaviors when linked with any components $B_2, B_4, B_5$
satisfying interfaces $i_2, i_4, i_5$.

Formally, full definedness is a language-specific parameter to our
definition of SCC, just as the program equivalence relations are
language-specific parameters to both SCC and vanilla full abstraction.
For instance, in the simple imperative language in
\autoref{sec:instance}, we will say that components $\ii{Cs}$
are fully defined with respect to a set of adversary interfaces
$\ii{BIs}$ if, for all components $\ii{Bs}$ 
satisfying $\ii{BIs}$, the complete program
$\link{\comp{\ii{Cs}} \mathrel{\cup} \comp{\ii{Bs}}}$ cannot reduce
to a stuck non-final state (corresponding to undefined behavior) where
the currently executing component is one of the ones in
$\ii{Cs}$ (\IE no component in $\ii{Cs}$ can be
``blamed''~\cite{FindlerF02prime} for undefined behavior).
Full definedness might well be defined differently for another
language; for instance, in a concurrent language undefined behaviors
cannot be as easily modeled by stuckness since normally other threads
can proceed even if one of the threads is stuck.
%
%
One last thing to note is that full definedness of a set of components is
generally a much weaker property than the full definedness of each
individual component in the set. Since the interfaces of
the adversary components \ii{BIs} can (and in \autoref{sec:instance} do)
restrict not only the operations they export but also the operations they
import from \ii{Cs}, the components in the set can export dangerous
operations just to other components in the set; the components actually in
the set might then all use these operations properly, whereas arbitrary
components with the same interfaces could abuse them to trigger undefined
behaviors.


\begin{figure}
\centering
\includegraphics[width=0.8\linewidth]{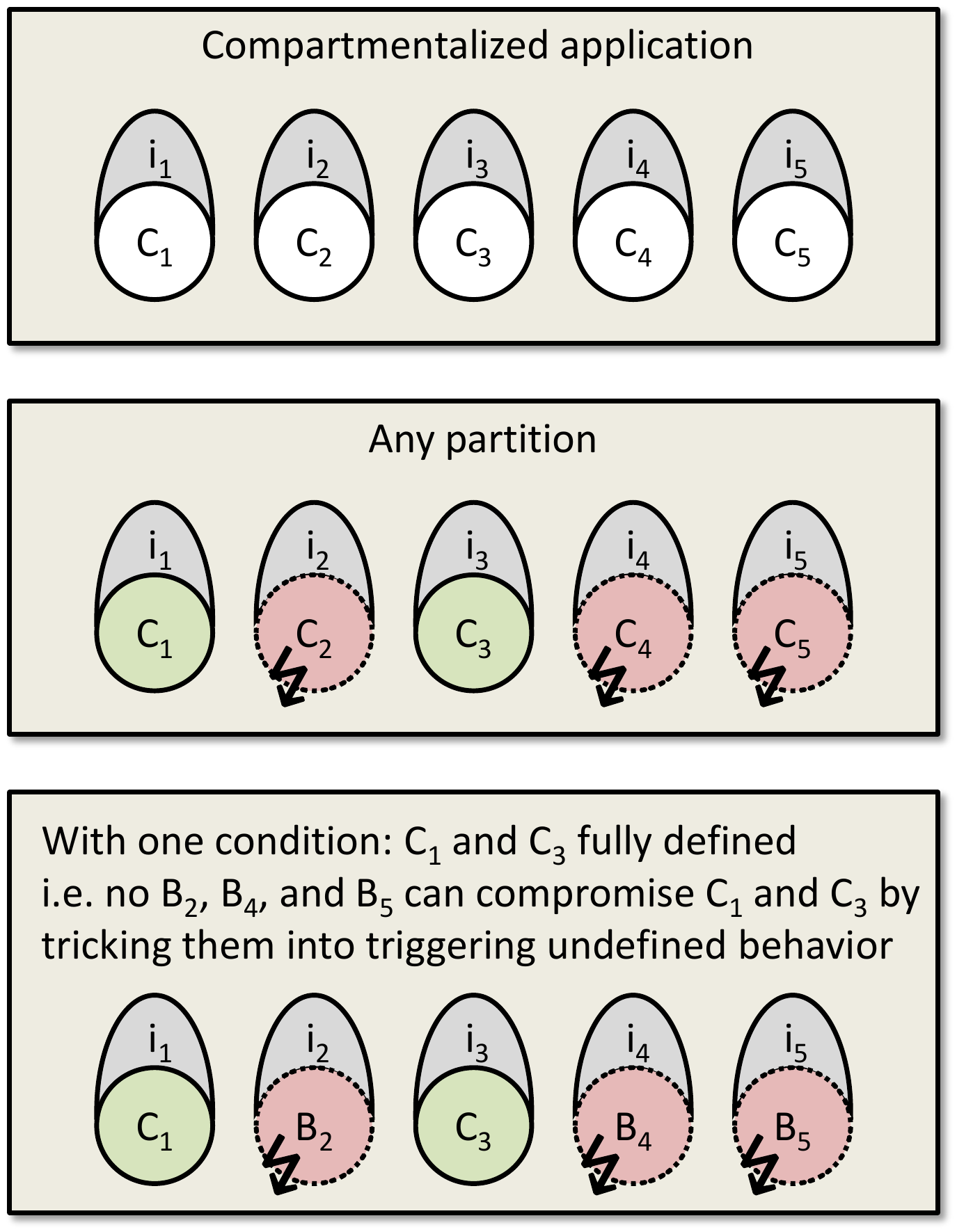}
\caption{Compromise scenarios}
\label{fig:compromise-scenario}
\end{figure}
\begin{figure}[t!]
\centering
\includegraphics[width=0.8\linewidth]{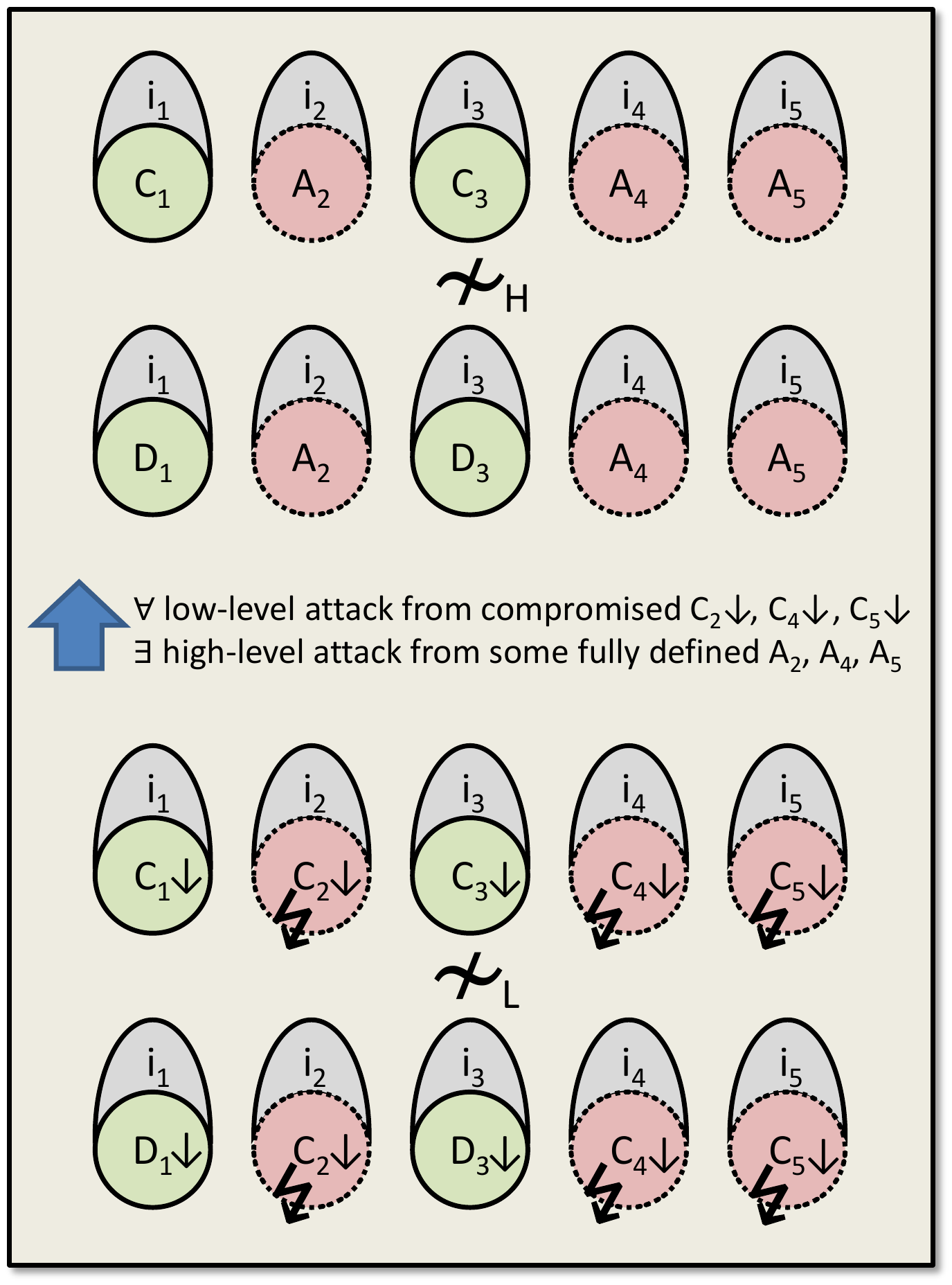}
\caption{SCC distinguishability game, for one of the
  compromise scenarios}
\label{fig:secure-compartmentalization-more-formal}
\end{figure}

SCC states that, \emph{in all such
  compromise scenarios}, the compiled compromised components must not be
able to cause more harm to the compiled uncompromised components via
low-level attacks than can be caused by some high-level components written
in the source language.
Basically this means that any low-level attack can
be mapped back to a high-level attack by compromised components satisfying
the given interfaces.
The property additionally ensures that the high-level components
produced by this ``mapping back'' are fully defined with respect to
the interfaces of the uncompromised components.
So with SCC,
instead of having to reason about the low-level consequences of
undefined behavior in the compromised components, we can
reason in the source language and simply replace the compromised
components by equivalent ones that are guaranteed to
cause no undefined behavior.

Formally, SCC is stated by quantifying over multiple
distinguishability games, one for each
compromise scenario, where the individual games are reminiscent of full
abstraction.
The goal of the attacker in each game is to distinguish between two
variants of the uncompromised components.
\autoref{fig:secure-compartmentalization-more-formal} illustrates
these two variants as $C_1, C_3$ and $D_1, D_3$, where we use $\neqh$ and
$\neql$ to indicate that the behaviors of two (high- or
low-level) complete programs are distinguishable, \IE they produce
different observable outcomes when executed.
For this compromise scenario, SCC specifies that,
if compiled compromised components $\comp{C_2}$, $\comp{C_4}$,
$\comp{C_5}$ can distinguish the $\comp{C_1},\comp{C_3}$ and
$\comp{D_1},\comp{D_3}$ variants at the low 
level, then 
there must exist some (fully defined) components $A_2, A_4, A_5$
that distinguish ${C_1},{C_3}$ and ${D_1},{D_3}$ at
the high level.

With all this in mind, the SCC property is
formally{}
expressed as follows:

\begin{defn}[SCC]~
\begin{itemize}
\item For any complete compartmentalized program and
for all ways of {\em partitioning} this program into
a set of {\em uncompromised} components $\ii{Cs}$
and their interfaces $\ii{CIs}$,
and a set of {\em compromised} components $\ii{Bs}$
and their interfaces $\ii{BIs}$, so that
$\ii{Cs}\text{ is {fully defined} with respect to }\ii{BIs}$, and 
\item for all ways of replacing the uncompromised components with
  components $\ii{Ds}$ that satisfy the same interfaces $\ii{CIs}$
  and are {fully defined} with respect to $\ii{BIs}$, 
\item if $\link{\comp{\ii{Cs}} \mathrel{\cup} \comp{\ii{Bs}}}
    \neql \link{\comp{\ii{Ds}} \mathrel{\cup} \comp{\ii{Bs}}}$,
\item then there exist components $\ii{As}$ satisfying interfaces $\ii{BIs}$
  and {fully defined} with respect to $\ii{CIs}$ such that\\
         $\link{\ii{Cs} \mathrel{\cup} \ii{As}}
    \neqh \link{\ii{Ds} \mathrel{\cup} \ii{As}}$.
\end{itemize}
\end{defn}

As suggested before, our property applies to any {\em fully defined}
sets of components \ii{Cs} and \ii{Ds} (which cannot be dynamically
compromised by some components with interfaces \ii{BIs}).
We conjecture that this full definedness precondition is strictly
required in the static corruption model we are assuming.
It is worth noting that we are not proposing any method for proving that
programs are fully 
defined; this comes with the territory when dealing with
C-like languages.
What we are after is bringing formal foundations
to {\em conditional} reasoning of the form ``if these \ii{Cs} are fully
defined and the remaining components \ii{Bs} get compromised,
then...''

Note that the $\ii{Bs}$ in our SCC definition need not be fully defined---\IE the
property allows the compromised components to contain undefined
behaviors (this may well be why they are compromised!) and promises
that, even if they do, we can find some other components $\ii{As}$
that are able to distinguish between $\ii{Cs}$ and $\ii{Ds}$ in
the source language without causing any undefined behaviors.
Indeed, for those compromise scenarios in which $\ii{Bs}$ are already
fully defined, our SCC property trivially follows from correct
compilation (\autoref{assumption:compiler-correctness}) since in that case
we can always pick $\ii{As} = \ii{Bs}$.

This generic property is parameterized over
  a source and a target language with a notion of component for each,
  source- and target-level notions of linking sets of components ($\bowtie$),
  source- and target-level notions of distinguishability ($\not\sim$),
  a compiler mapping source components to target components ($\downarrow$),
  a source-level notion of interface and an interface satisfaction relation
  (lifted
  to sets of components and interfaces), and
  a notion of a set of components $\ii{Cs}$ being
  fully defined with respect to a set of adversary interfaces $\ii{BIs}$.

\section{From Full Abstraction to SCC}
\label{sec:fa-not-enough}

\autoref{sec:sc} presented SCC by directly
characterizing the attacker model against which it defends.
In this section we step back and show how SCC can instead be obtained by
starting from the well-established notion of full abstraction and
removing each of the three limitations that make it unsuitable
in our setting.
This results in two properties, {\em structured full abstraction} and {\em
  separate compilation}, which we then combine and instantiate to obtain SCC.
This reduction is
not only theoretically interesting, but also practically useful,
since structured full abstraction can more easily be proved by adapting
existing proof techniques, as we will see in \autoref{sec:instance}.

%

\SUBSECTION{Full abstraction}
\label{sec:fa}

A {\em fully abstract} compiler protects compiled programs from their interaction
with unsafe low-level code and thus allows sound reasoning about security
(and other aspects of program behavior) in terms of the source language.
Fully abstract compilation~\cite{abadi_protection98}
intuitively states that no
low-level attacker can do more harm to a compiled program than some
program in the source language already could.
This strong property requires enforcing all high-level language
abstractions against arbitrary low-level attackers.

%
Formally, full abstraction is phrased as a distinguishability game
requiring that low-level attackers have no more distinguishing
power than high-level ones.

\begin{defn}\label{defn:fa-simple}
We call a compilation function (written $\downarrow$) {\em fully abstract}
if, for all $P$ and $Q$,
\[
(\forall A.~ A[P] \eqh A[Q])
\Rightarrow (\forall a.~ a[\comp{P}] \eql a[\comp{Q}]).
\]
\end{defn}

\noindent Here, $P$ and $Q$ are partial programs, $A$ is a high-level
context whose job is to try to distinguish $P$ from $Q$, and $a$ is a
low-level ``attacker context'' that tries to distinguish $\comp{P}$ from
$\comp{Q}$.
The relations $\eql$ and $\eqh$ are parameters to the definition,
representing behavioral equivalence at the two levels.
To be useful, they should allow the context to produce an observable action
every time it has control, letting it convert its knowledge into
observable behaviors.
For instance, a common choice for behavioral equivalence is based on
termination: two deterministic programs are behaviorally equivalent if
they both terminate or both diverge.

When stated this way (as an implication rather than an equivalence),
full abstraction is largely orthogonal to compiler
correctness~\cite{leroy09:compcert, KumarMNO14}.
While compiler correctness is about preserving behaviors when
compiling from the source to the target, proving full abstraction
requires some way to map each distinguishing context target to a
sourge-level one, which goes in the opposite direction.
This is easiest to see by looking at the contrapositive:
\[
\forall a.~ a[\comp{P}] \neql a[\comp{Q}] \Rightarrow
\exists A.~ A[P] \neqh A[Q]
\]



\SUBSECTION{Problem 1: Undefined behavior}
\label{sec:prob1}

The first limitation of full abstraction is that it cannot realistically be
applied to compiling from an unsafe language with undefined behaviors.
%
%
Undefined behaviors are (arbitrarily!) nondeterministic, and no
realistic compiler can preserve this nondeterminism in the target
as required by full abstraction. (Removing it from the source language
would negate the performance and optimization benefits that are the
reason for allowing undefined behaviors in the first place.)


To adapt full abstraction to a source language with undefined behaviors, we
need to restrict attention only to {\em defined}
complete programs in the source language.
%
%
%
And even with this restriction, defining full abstraction still
requires a little care. For instance, the following variant is
wrong (formally, {\em defined} is another parameter to this property):%
\[
\begin{array}{rl}
(\forall A.~ A[P]\text{ and }A[Q]\textit{ defined }\Rightarrow A[P] \eqh A[Q])
&\Rightarrow\\
(\forall a.~ a[\comp{P}] \eql a[\comp{Q}])
\end{array}
\]
Any $P$ and $Q$ that both trigger undefined behavior as soon
as they get control would be considered equivalent in the high-level language
because there is no context that can make these programs defined
while observing some difference between them.
All such programs would thus need to be equivalent at the low level,
which is clearly not the case (since their nondeterminism can be resolved
in different ways by the compiler).
The problem here is that if $P$ and $Q$ trigger undefined behavior
then the context often cannot make up for that and make the program
defined in order be able to cause an observation that distinguishes
$P$ and $Q$.

\SUBSECTION{Solution 1: Full abstraction for unsafe languages}
\label{sec:lfa}


The responsibility of keeping $A[P]$ defined should be thus shared
between $A$ and $P$.
For this we assume a compositional notion of {\em fully defined}
behavior for programs and contexts as two parameters to
\autoref{defn:lfa} below.
We require that these parameters satisfy the following properties: (1)
a program is fully defined if it does not cause undefined behavior in
any fully defined context, and (2) a context is fully defined if it
does not cause undefined behavior when we plug any fully defined
program into it.\ch{Above, should we have ``if'' or ``iff''?}
Note that properties (1) and (2) are circular and therefore cannot be
used as the definition of full definedness.
For specific languages (\EG the one in \autoref{sec:instance}) we can
break this circularity and define full definedness using {\em
  blame}~\cite{FindlerF02prime}: intuitively we call a partial program
{\em fully defined} when it cannot be blamed for undefined behavior in
any context whatsoever.
Similarly, we call a context fully defined when it cannot be blamed
for undefined behavior for any program that we plug into it.
%
We expect such a blame-based definition to satisfy the properties (1)
and (2) above.
Full definedness allows us to introduce a new variant of
full abstraction that applies to unsafe source languages with
undefined behavior:

\begin{defn}[Full abstraction for unsafe languages]\label{defn:lfa}~\\
  We call a compiler $\downarrow$ for an unsafe language {\em fully
    abstract} if for all {\em fully defined} partial programs $P$ and
  $Q$%
\[
\begin{array}{rl}
(\forall A.~ A~\textit{fully defined} \Rightarrow A[P] \eqh A[Q])&\Rightarrow\\
(\forall a.~ a[\comp{P}] \eql a[\comp{Q}]).
\end{array}
\]
\end{defn}

Requiring that $P$, $Q$, and $A$ are fully defined means that we can safely
apply $\eqh$ to $A[P]$ and $A[Q]$, because neither the programs nor
the context can cause undefined behavior.
This property is incomparable with the original definition of full
abstraction.
Looking at the contrapositive,
\[
\begin{array}{r}
\forall P, Q\textit{ fully defined}.\qquad
   (\exists a.~ a[\comp{P}] \neql a[\comp{Q}]) \\
   \Rightarrow(\exists A.~ A~\textit{fully defined} \wedge A[P] \neqh A[Q]),
\end{array}
\]
the $P, Q\textit{ fully defined}$ pre-condition makes this weaker than
full abstraction, while the $A~\textit{fully defined}$ post-condition makes it stronger.
%
%
The post-condition greatly simplifies reasoning about programs 
by allowing us to replace reasoning about low-level contexts with
reasoning about high-level contexts {\em that cannot cause undefined
behavior}.

One might wonder whether the
$P, Q\textit{ fully defined}$ pre-condition is too restrictive,
since full definedness is a rather strong property, requiring each 
component to be very defensive about validating inputs it receives from
others.
In the static compromise model inherent to full abstraction and
without additional restrictions on the program's context, we must 
be conservative and assume that, if any context can cause
undefined behavior in a program, it can compromise it in a way that the
compiler can provide no guarantees for this program.
The structured full abstraction definition below will in fact restrict
the context and thus use a weaker notion of full definedness.
Moreover, separate compilation will allow us to quantify over all splits of a
program into a fully defined partial program and a compromised
context, which also makes the presence of the full definedness
pre-condition more palatable.

\SUBSECTION{Problem 2: Open-world assumption about contexts}
\label{sec:prob2}





While full abstraction normally requires the contexts
to be compatible with the partial program, for instance by respecting
the partial program's typed interface,
these restrictions are minimal
and do not restrict the shape%
, size,
exported interface, or privilege
of the contexts in any way.
This {\em open world} assumption about contexts does not fit with our
compartmentalization setting, in which the breakup of the
application into components is fixed in advance, as are the
interfaces (and thus privileges) of all the components.
In our setting, the definition of full abstraction needs to be refined
to track and respect such structural constraints; otherwise a
low-level context with 2 components might be mapped back to a
high-level context with, say, 3 components that have completely different
interfaces, and thus privileges.
In particular, the high-level components' interfaces could give them
more privileges than the low-level components had, 
increasing their distinguishing power.

\SUBSECTION{Solution 2: Structured full abstraction}
\label{sec:sfa}

We therefore introduce a structured variant of full abstraction, in which
partial programs (indicated by $\bullet$ below) and contexts ($\circ$) are
assigned dual parts of predefined complete program {\em shapes}.
A shape might be anything, from a division into components with their
interfaces (as in \autoref{thm:sfa-to-sc} below), to, \EG the maximum
size of a component's code after compilation (which might expose component
sizes in a setting where it's too costly to hide them by padding to a fixed
maximum size~\cite{PatrignaniDP16}).
%
%

\begin{defn}[Structured full abstraction]\label{defn:sfa}~\\
  We say that a compiler $\downarrow$ for an unsafe language
  satisfies {\em structured full abstraction} if, for all
  {\em program shapes} $s$ and partial programs
    $P \hasshape{\bullet} s$ and $Q \hasshape{\bullet} s$ so that
    $P$ and $Q$ are {\em fully defined} with respect to contexts of shape
    ${\circ}s$,
\[
\begin{array}{r}
\left(\begin{array}{r}
\forall A \hasshape{\circ} s.\quad
   A~\textit{fully defined}\text{ wrt. programs of shape }{\bullet}s\\
   \qquad\Rightarrow A[P] \eqh A[Q]
\end{array}\right)\\[1em]
\Rightarrow(\forall a \hasshape{\circ} s.~ a[\comp{P}] \eql a[\comp{Q}])
.
\end{array}
\]
\end{defn}

\noindent This property universally quantifies over any
complete program shape $s$ and requires that $P \hasshape{\bullet} s$
(read ``program $P$ has shape $s$''), $Q \hasshape{\bullet} s$, and
$A \hasshape{\circ} s$ (``context $A$ matches programs of shape
$s$'').
Moreover, the property only requires programs that are fully defined
with respect to contexts of the right shape, and dually it only
considers contexts that are fully defined with respect to programs of
the right shape.

\SUBSECTION{Recovering secure compartmentalizing compilation}
\label{sec:recovering-scc}

SCC can be recovered in a natural
way as an instance of structured full abstraction (\autoref{defn:sfa}).
%
For both source and 
target languages, we take partial programs and contexts be sets of
components and context application be set union.
Compilation of sets of components works pointwise.
%
%
To obtain an instance of structured full abstraction we additionally
take shapes to be sets of component interfaces, where each interface
is marked as either compromised or uncompromised.  

\begin{thm}\label{thm:sfa-to-sc}
  For any deterministic target language and any source language that
  is deterministic for defined programs, structured full abstraction
  instantiated to components as described above implies SCC.
\end{thm}
\begin{proof}
Straightforward, though tedious.  A machine-checked Coq proof can be found
in the auxiliary materials.\Coqed
\end{proof}

\SUBSECTION{Problem 3: Statically known trusted/untrusted split}

While SCC can deal with multiple
compromise scenarios, not all instances of structured full abstraction
can.
In general, if a compiler satisfies (structured) full abstraction,
how can we know whether it can deal with multiple compromise scenarios,
and what does that even mean?
While we can instantiate full abstraction to a
{\em particular} compromise scenario by
letting the partial program $P$ contain the uncompromised components and the
low-level context $a$ contain the compromised ones, a
fully abstract compiler
(together with its linker, loader, runtime \ETC) might exploit this
static split and introduce only one single barrier protecting the
uncompromised components from the compromised ones.
When presented with a different compromise scenario for the same program,
the compiler could adapt and produce a different output.

The source of confusion here is that a fully abstract compiler does not need
to compile contexts---only programs. In fact, even the {\em types} of
contexts and of partial programs might well be completely different (\EG the
types of lambda calculus contexts and terms are different; a compiler for
one cannot compile the other).
Even when the types do match so that we can apply the same
compiler to the context, the low-level context-application operation
$\comp{A}\hspace{-0.35em}[\comp{P}]$ can freely exploit the fact
that its first argument is a compiled untrusted context and its second
argument is a compiled trusted program that should be protected from
the context.
So if we start with a complete high-level program $C$ and
look at two different compromise scenarios $C = A_1[P_1]$ and
$C = A_2[P_2]$, compiling each of the parts and combining the results
using context application does not necessarily yield the same result
(\IE it could well be that
$\comp{A_1}\hspace{-0.35em}[\comp{P_1}] \not=
\comp{A_2}\hspace{-0.35em}[\comp{P_2}]$) or indeed even behaviorally
equivalent results (\IE it could well be  that
$\comp{A_1}\hspace{-0.35em}[\comp{P_1}] \neql
\comp{A_2}\hspace{-0.35em}[\comp{P_2}]$).
This means that the user of a fully abstract compiler may need to commit
{\em in advance} to a single compromise scenario.

This weakness significantly limits the applicability of full abstraction.
After all, uncertainty about sources of vulnerability is precisely the
motivation for compartmentalizing compilation: if we knew 
which components
were safe and which were not, there would be no reason to distinguish more
than two levels of privilege, and we could merge each group into a single
mega-component.
Even in rare cases where we are certain that some code cannot be
compromised---for instance because we have proved it
safe---protecting only the verified code from all the rest using a fully
abstract compiler~\cite{Agten0P15} is still suboptimal in terms of
protection, since it provides no guarantees for all the code that is not
verified.

Moreover, this weakness is not hypothetical:
several fully abstract compilers proposed in the literature are only
capable of protecting a single trusted module from its untrusted
context~\cite{PatrignaniASJCP15, AgtenSJP12, LarmuseauPC15,
  PatrignaniCP13, patrignani_thesis} (recently proposed
extensions~\cite{PatrignaniDP16} do aim at lifting this restriction in
some cases).
While this setup is appropriate when all one wants to achieve is
protecting trusted (e.g., verified) code
from its untrusted context~\cite{Agten0P15}, it is not suitable for
our compartmentalizing compilation setting where we do not know in advance
which components will be dynamically compromised and which ones not,
and thus we want to simultaneously protect against all possible compromise
scenarios.

\SUBSECTION{Solution 3: Separate compilation}
\label{sec:separate-compilation}

We can address this by requiring that the compiler toolchain have one
additional property: 

\begin{defn}\label{defn:separate-compilation}
  We say that the compiler toolchain (\IE the compiler $\comp{-}$, the linker
  $-[-]$, and the runtime system embodied in the low-level behavioral
  equivalence) satisfies {\em separate compilation} if
\begin{enumerate}
\item the type of contexts and programs is the same (so that the
  compiler can also compile contexts), and
\item $\comp{(A[P])} \eql{}\; \comp{A}\hspace{-0.35em}[\comp{P}]$ for all
$A$ and $P$.
\end{enumerate}
\end{defn}


Requiring that context application and compilation commute (condition 2)
implies that, if some complete program $C$ can be written
as both $C = A_1[P_1]$ and $C = A_2[P_2]$, then separately compiling
each of these splits yields behaviorally equivalent results:
$\comp{(A_1[P_1])} \eql\; \comp{(A_2[P_2])}$.
With separate compilation, full abstraction for an unsafe language
(\autoref{defn:lfa}) can be instantiated as follows:
\[
\begin{array}{r}
\forall B.~\forall P,Q\textit{ fully defined}.\qquad
(\comp{(B[P])} \neql \comp{(B[Q])})\\
\Rightarrow (\exists A.~ A\textit{ fully defined} \wedge A[P] \neqh A[Q])
\end{array}
\]
One compelling reading of this is that, for all compromise scenarios
(ways to break a complete program into a compromised context $B$ and
an uncompromised program $P$),
and for all programs $Q$ that we can substitute for $P$, if the
context $B$ can distinguish $P$ from $Q$ when compiled to
low-level code, then there exists a fully defined context $A$
that can distinguish them at the high-level.

In a language without undefined behavior, this property would
trivially follow just from (whole program) correct compilation (see
\autoref{assumption:compiler-correctness} below) by picking $A = B$.
However, it is nontrivial for a language in which
context $B$ might cause undefined behavior, since then correct
compilation does not apply for $B[P]$ and $B[Q]$.
In our setting, this property allows us to avoid reasoning about the
implications of undefined behavior in a low-level
context and instead consider just fully defined high-level
contexts. 

It is trivial to check that our instance of structured full
abstraction from \autoref{thm:sfa-to-sc} does satisfy separate
compilation.
It should also be easy to show that many previous fully abstract
compilers~\cite{PatrignaniASJCP15, AgtenSJP12, LarmuseauPC15,
  PatrignaniCP13, patrignani_thesis} do not satisfy separate
compilation, since they were not designed to support a setting of
mutual distrust.








\section{A Simple Instance of SCC}
\label{sec:instance}

In this section, we illustrate the main ideas behind SCC with a
proof-of-concept compiler from an unsafe language with components to
an abstract machine with compartments.
We discuss key design decisions for providing secure
compartmentalization, such as cleaning register values to prevent
unintended communication between compartments.
We also explain how a compiler optimization for component-local calls
makes unwary compiled components vulnerable to realistic control-flow
hijacking attacks.
Finally, we show how to adapt a standard full abstraction proof
technique called {\em trace semantics}~\cite{PatrignaniASJCP15} to
prove SCC.

The results in this section have been proved on paper under
the assumptions explicitly mentioned in the text. 
In the following, an assumption denotes a property that we believe
is true and rely on, but that we haven't proved.
Lemmas, theorems and corollaries denote properties that we have
proved on paper, possibly relying on some of the stated assumptions.
The proof of the structured full abstraction theorem
(\autoref{thm:sfa-instance}) has also been formalized in Coq assuming
most other results in this section as axioms, with the exception of
\autoref{thm:partial-type-safety} and
\autoref{cor:separate-compiler-correctness} which are also mechanized.
While not constituting a complete machine-checked proof, this
mechanization further validates the high-level proof structure
described in this section.

\SUBSECTION{Source Language}
\label{sec:source}

We work with an unsafe source language with components,
procedures, and buffers.  A program in this language is a set
of components communicating via procedure calls.
Buffer overflows have undefined behavior and may open the door to
low-level attacks after compilation.
However, thanks to the cooperation between the low-level
compartmentalization mechanism and the compiler,
the effects of these attacks will be limited to the offending
component.

Components have statically checked interfaces that specify which
procedures they import and export.
To satisfy an interface, a component can only call external procedures
that it imports explicitly, and it must define all procedures exported
by its interface.
Thus, interfaces define privileges by preventing components from
calling non-imported procedures, and enable components to define
private procedures (that are not exported in their interfaces).
We will use the same notion of interfaces in our target abstract machine.

The syntax of expressions, given below, is that of a standard
imperative language with mutable buffers and mutually recursive
procedures.  Each component $C$ has local procedures \iffull
\else ``$C.P$'' \fi and private
local buffers \iffull
\else $b$\fi{}.
Loops are encoded using recursive calls, sequencing is encoded as a
binary operation, and variables are encoded using buffers.  Procedures
take a single argument, which by convention is always passed in the
first cell of the first buffer of the callee component.  The only
first class values are integers $i$; these can be passed across
component boundaries using procedure calls and returns.  Buffers and
procedures are second class.
\[
\begin{array}{l@{}l}
  \mathit{e} \ \ \mathrel{::=}\ \  \;&
    i
    \;|\;
    e_1 \otimes e_2
    \;|\;
    \texttt{if }e\texttt{ then }e_1\texttt{ else }e_2
    \;|\;
    b\texttt{[}e\texttt{]}
    \;|\;
  \\  &
    b\texttt{[}e_1\texttt{] := }e_2
    \;|\;
    C\texttt{.}P\texttt{(}e\texttt{)}
    \;|\;
    \texttt{exit}
\end{array}
\]
where $\otimes \in \{ ; , + , - , \times, =, \leq, \ldots\}$.

We define a standard continuation-based small-step semantics that
reduces configurations $\mathit{cfg}$.  It is
deterministic for programs that don't cause undefined behavior.
\aaa{This is confusing, since source-level undefined behavior causes
  stuckness, and not nondeterminism. Especially given the discussion
  on stuckness a few paragraphs later.}\ch{+1, how about just saying:
  it is deterministic; undefined behaviors are modeled as stuckness
  as discussed below. Otherwise I would remove this from here and
  only discuss determinism later, together with undefined behaviors.}
\[
\begin{array}{l@{}l@{~~~~~}l@{}l}
  \mathit{cfg} \ \mathrel{::=}\ \;& (C, s, \sigma, K, e)
  &
    \mathit{K} \ \mathrel{::=}\ \;&
    \texttt{[]}
    \;|\;
    E \texttt{::} K
\end{array}
\]
\[
\begin{array}{l@{}l}
  \mathit{E} \ \ \mathrel{::=}\ \  \;&
    \square{} \otimes e_2
    \;|\;
    i_1 \otimes \square{}
    \;|\;
    \texttt{if }\square{}\texttt{ then }e_1\texttt{ else }e_2
    \;|\;
    \\ &
    b\texttt{[}\square{}\texttt{] := }e_2
    \;|\;
    b\texttt{[}i_1\texttt{] := }\square{}
    \;|\;
    C\texttt{.}P\texttt{(}\square{}\texttt{)}
\end{array}
\]
A configuration $(C, s, \sigma, K, e)$ represents a call in progress
within component $C$, in which $e$ is the expression being reduced and
$K$ is the continuation for this expression, up to the latest procedure
call.
Continuations are evaluation contexts, here represented as lists of
flat evaluation contexts $E$.
We denote by $s$ a global state recording the values of the local
buffers for each component.
Continuations for pending calls are stored on a call stack
$\sigma$, together with their call arguments' values and the names of
the compartments they execute in.
We omit the obvious definitions for call stacks $\sigma$ and states $s$.

Evaluation starts as a call to a fixed procedure of a fixed main
component, and completes once this call completes, or whenever the
current expression $e$ is $\texttt{exit}$.
We illustrate the small-step semantics with the three rules that deal
with procedure call evaluation.
In these rules, $\Delta$ is a mapping from procedure identifiers to
procedure bodies.

\infrule[]
  {s' = s[C',0,0 \mapsto i] \andalso \sigma' = (C,s[C,0,0],K) \texttt{::} \sigma }
  {\Delta \vdash (C,  s, \sigma, C'.P'(\square) \texttt{::} K, i) \rightarrow
   (C', s', \sigma', \texttt{[]}, \Delta[C',P'])}

\infrule[]
 {s' = s[C',0,0 \mapsto i']}
 {\Delta \vdash (C, s, (C',i',K) \texttt{::} \sigma, \texttt{[]}, i) \rightarrow (C', s', \sigma, K, i)}

\infrule[]
  {}
  {\Delta \vdash (C, s, \sigma, K, C'.P'(e)) \rightarrow (C, s, \sigma, C'.P'(\square) \texttt{::} K, e)}

As shown on the right-hand side of the first rule, a call starts with
an empty continuation and the procedure body $\Delta[C',P']$ as the
current expression.
The first cell in the first buffer of the callee compartment is
updated with the call argument, while information about the caller's
state when performing the call gets pushed on the call stack $\sigma$.
A call completes once an empty continuation is reached and the current
expression is a value, as is the case in the left-hand side
of the second rule.
In this case, the caller's state is restored from the call stack, and
execution resumes with the call result $i$ as the current expression.
The intermediate steps between the start and the end of a call reduce
the procedure body to a value, as the last rule illustrates:
Whenever $e$ is not a value, reduction deconstructs $e$ into a
subexpression $e'$ and a flat evaluation context $E$ such that
$e = E[e']$, where $E[e']$ means filling the hole $\square$ in $E$
with $e'$.
This expression $e'$ becomes the new currently reduced expression,
while $E$ gets appended on top of the current continuation $K$.
Finally, when $e$ is a value $i$ and the call has not completed
($K \neq \texttt{[]}$), the next step is chosen
based on the flat evaluation context found on top of $K$,
which gets removed from $K$.
In the left-hand side of the first rule, for example, this flat
evaluation context is $C'.P'(\square)$, for which the next chosen
step, as shown on the right-hand side, is to start a procedure call to
$C'.P'$, using $i$ as the call argument.

Since undefined behaviors are allowed to
take the machine to an arbitrary {\em low-level} state, 
it wouldn't make much sense to try to make the source-language semantics
describe what can happen if an undefined point is reached.  
We therefore model them at the source level simply as
stuckness (as done for instance in CompCert~\cite{Leroy09}).
In particular, reduction gets stuck when trying to
access or update a buffer out of bounds, and 
the type safety theorem says
that well-formed programs can only go wrong (get stuck) by reducing to an
out-of-bounds operation on a buffer.
A program is well-formed if all the used buffers are defined, all
imported components are defined, all imported external procedures are
public, and if the names of all components are unique.
Well-formedness extends straightforwardly to configurations.
%

\begin{thm}[Partial type safety]\Coqed
\label{thm:partial-type-safety}
For any well-formed configuration $\textit{cfg} = (C, s, \sigma, K, e)$,
one of the following holds:
\begin{enumerate}[(1)]
\item $\textit{cfg}$ is a final configuration (either
$e$ is  
$\texttt{exit}$ or else it is a value and $K$ and $\sigma$ are both
empty);
\item \textit{cfg} reduces in one step to a well-formed configuration;
\item \textit{cfg} is stuck and has one of the following forms:
 \begin{enumerate}[(a)]
   \item $(C, s, \sigma, b\texttt{[}\square\texttt{]} :: K, i)$
   where $s[C,b,i]$ is undefined;
   \item $(C, s, \sigma, b\texttt{[}i\texttt{]:=}\square {::} K, i')$ where
   $s[C,b,i]$ is undefined.
 \end{enumerate}
\end{enumerate}
\end{thm}

In the following, we use the term \emph{undefined behavior configurations}
for the
configurations described in (3), and we say that a well-formed program
is \emph{defined} if reducing it never reaches an undefined behavior
configuration.\ch{This now seems a bit out of sync with claims in section
  2 that we model undefined behavior as stuckness. No big deal though,
  since for well-defined programs the two agree.}


\SUBSECTION{Target}
\label{sec:target}

Our compiler targets a RISC-based abstract machine extended with a
compartmentalization mechanism, inspired by a similar design featured
in previous work~\cite{micropolicies2015}. Each compartment in this
target has its own private memory, which cannot be directly accessed
by others via loads, stores, or jumps. Instead, compartments must
communicate using special call and return instructions, which, as
explained below, include checks to ensure that the communication
respects compartment interfaces. (Note that this scheme requires a
protected call stack, which, in a real system, could be implemented
\EG using a shadow call stack~\cite{ErlingssonAVBN06, AbadiBEL09} or
return capabilities~\cite{JuglaretHAPST15}.)
Because resource exhaustion and integer overflow issues are orthogonal
to our present concerns, we assume that words are unbounded and
memory is infinite.

The instruction set for our machine is mostly standard.
\[
  \begin{array}{l@{}l}
  \mathit{instr} \ \ \mathrel{::=}\ \ \;&
    \con{Nop}
    \;|\;
    \con{Const}~i\rightarrow{}r_d
    \;|\;
    \con{Mov}~r_s\rightarrow{}r_d
  \\ &
    \;|\;
    \con{Load}~^{*}r_p\rightarrow{}r_d
    \;|\;
    \con{Store}~^{*}r_p\leftarrow{}r_s
  \\ &
    \;|\;
    \con{Jump}~r
    \;|\;
    \con{Jal}~r
    \;|\;
    \con{Call}~C~P
    \;|\;
    \con{Return}
  \\ &
    \;|\;
    \con{Binop}~r_1\otimes{}r_2\rightarrow{}r_d
    \;|\;
    \con{Bnz}~r~i
    \;|\;
    \con{Halt}
  \end{array}
\]
$\con{Const}~i\rightarrow{}r_d$ puts an immediate value $i$ into register $r_d$.
$\con{Mov}~r_s\rightarrow{}r_d$ copies the value in $r_s$ into $r_d$.
$\con{Load}~^{*}r_p\rightarrow{}r_d$ and
$\con{Store}~^{*}r_p\leftarrow{}r_s$ operate on the memory location whose
address is stored in $r_p$ (the $*$ in the syntax of \con{Load} and
\con{Store} indicates that a pointer dereference is taking place), either
copying the value found at this 
location to $r_d$ or overwriting the location with the content of $r_s$.
$\con{Jump}~r$ redirects control flow to the address stored in $r$.
$\con{Jal}~r$ (jump-and-link) does the same but also communicates a
return address in a dedicated register $r_{ra}$, so that the target
code can later resume execution at the location that followed the
$\con{Jal}$ instruction.
$\con{Call}~C~P$ transfers control to compartment $C$ at the entry
point for procedure ``$C.P$''.
$\con{Return}~C~P$ transfers control back to the compartment that
called the current compartment.
$\con{Binop}~r_1\otimes{}r_2\rightarrow{}r_d$ performs the
mathematical operation $\otimes$ on the values in $r_1$ and $r_2$ and
writes the result to $r_d$.
Finally, $\con{Bnz}~r~i$ (branch-non-zero) is a conditional branch
to an offset $i$, which is relative to the current program counter.
If $r$ holds anything but value zero, the branching happens, otherwise
execution simply flows to the next instruciton.

While $\con{Jal}$ is traditionally used for procedure calls and
$\con{Jump}$ for returns, in this machine they can only target the
current compartment's memory.
They are nonetheless useful for optimizing compartment-local calls, which
need no instrumentation; in a realistic setting, the instrumented
primitives $\con{Call}$ and $\con{Return}$ would likely come with
monitoring overhead.


In the low-level semantics, we represent machine states $\ii{state}$
as
$(C,\sigma,\ii{mem},\ii{reg},\ii{pc})$ where $C$ is the currently
executing compartment, $\ii{mem}$ is a partitioned memory, $\ii{reg}$
is a register file, $\ii{pc}$ is the program counter,
and $\sigma$ is a global protected call stack.
We assume a partial function $\ii{decode}$ from words to
instructions.
We write $\psi; E \vdash \ii{state} \to \ii{state'}$ to mean
that $\ii{state}$ reduces to $\ii{state'}$ in
an environment where component interfaces are given by $\psi$ and
component entry points by $E$.
Here are the reduction rules for $\con{Call}$ and $\con{Return}$:
\infrule[]
  {\ii{mem}[C,\ii{pc}] = i \andalso
  \ii{decode}~i = \con{Call}~C'~P' \andalso
  \ii{pc}' = E[C'][P']
    \\
   C' = C ~ \vee ~ C'.P' \in \psi[C].\ii{import} \andalso
   \sigma' = (C,\ii{pc}\mathord{+}1)~\con{::}~\sigma
   }{\psi; E \vdash (C,\sigma,\ii{mem},\ii{reg},\ii{pc}) \to
                    (C',\sigma',\ii{mem},\ii{reg},\ii{pc}')}
\infrule[]
  {\ii{mem}[C,\ii{pc}] = i \andalso
   \ii{decode}~i = \con{Return} \andalso
   \sigma = (C',\ii{pc}')~\con{::}~\sigma'
  }{\psi; E \vdash (C,\sigma,\ii{mem},\ii{reg},\ii{pc})
    \to (C',\sigma',\ii{mem},\ii{reg},\ii{pc}')}
The $\con{Call}$ rule checks that the call is valid with
respect to the current compartment's interface---i.e., the target
procedure is imported by the current compartment---which
ensures that even if a compiled component is compromised
it cannot exceed its static privilege level.
Then it puts the calling compartment's name and program
counter on the global protected call stack $\sigma$.
Finally, it redirects control to
the entry point of the called procedure.
The $\con{Return}$ instruction retrieves the caller's compartment and
return address from the protected call stack and resumes execution there.

\SUBSECTION{Compiler}
\label{sec:compiler}

We next define  a simple compiler that produces one low-level memory
compartment for each high-level component.
Each compartment is internally split into buffers,
the code of procedures, and a local stack that can grow infinitely.
The local stack is used to store both intermediate results and return
addresses. 

In standard calling conventions, the callee is generally expected to restore
the register values of the caller, if it has modified them, before
returning from a call.
Here, however, compiled components cannot assume that other
components will necessarily follow an agreed calling convention, so they
must save any register that may be needed later.
This means, for instance, that we save the value of the current
call argument on the local stack and write the
local stack pointer to a fixed location in the current compartment's
memory before any cross-compartment call instruction is performed,
so that the compartment can restore them when it gets control back.

The compiler must also prevent a compromised compartment from reading
intermediate states from code in other compartments that may be in the
middle of a call to this one.
Intuitively, a secure compiler must prevent compromised compartments from
distinguishing compiled components based on low-level information that
(fully defined) high-level attackers don't get.  In the source language,
only a single argument or return value is communicated at call and return
points.  
Hence, besides preserving their values for later, the compiler
ensures that all\footnote{
  Technically speaking, we believe that, in our very simple setting,
  the compiler could choose not to clean \emph{unused} registers and
  still be secure.
  However, our proof relies on compiled components cleaning \emph{all}
  registers except the one that holds the call argument or return
  value.
  Indeed, not cleaning unused registers makes things harder because it
  can provide a covert channel for two compromised compartments
  between which interfaces would forbid any direct communication.
  These compartments could now exchange values through uncleared
  registers by interacting with the same unsuspecting uncompromised
  component.
  We conjecture that this possible cooperation between compromised
  components doesn't yield more attacker power
  in our case.
  However, in a
  setting
  where registers could be used to transmit capabilities,
  this \emph{would} give more power to the attacker, so
  our compiler clears all registers but one, which also simplifies our
  proof.}
registers are cleaned before transferring control to other
compartments. 
%
%

The compiler implements a simple optimization for local calls.
Since all procedures of a component live in the same
address space and local calls don't need instrumentation, these calls
can be implemented more efficiently using $\con{Jal}$ and
$\con{Jump}$ instructions.
We therefore use different procedure entry points for
component-local and cross-component calls, and we skip, for
local calls, the steps that store and restore register values
and clean registers.

Because we do not check bounds when compiling buffer read and write
operations, buffer overflows can corrupt a compartment's memory in
arbitrary ways.
Consequently, many buffer overflow attacks can be reproduced even in our
simple setting,
including, due to the local-call optimization,
return-oriented programming attacks~\cite{Shacham07, Buchanan2008}.
In return-oriented programming, an attacker overwrites return
addresses on the local stack to produce an unexpected sequence of
instructions of his choice by reusing parts of the code of
component-local procedures.
In our setting, buffer overflow attacks thus enable compiled components to
shoot themselves in the foot by storing beyond the end of a buffer and into
the local call stack.

We assume compiler correctness as stated below for our
compiler.  
%
\iffull
\else{}
Note that, in the presence of partial type safety,
(\autoref{thm:partial-type-safety}), proving either (1)
or (2) below is enough to get the other.
\fi
\ch{Might want to make it explicit that this definition of compiler
  correctness relies on the determinism of both source and target, right?}


\begin{assumption}[Whole-program compiler correctness]\label{assumption:compiler-correctness}
\[
\begin{array}{l}
\forall P.~ P\text{ defined} \Rightarrow \\
   \text{~~(1) }\TERM{P} \iff \TERM{\comp{P}}~\wedge \\
   \text{~~(2) }\DIV{P} \iff \DIV{\comp{P}}
\end{array}
\]
\end{assumption}

\SUBSECTION{Instantiating structured full abstraction}
\label{sec:instance-defs}

We define program shapes, partial programs, and contexts in a similar
way to \autoref{thm:sfa-to-sc}, as detailed below.
More precisely, we use isomorphic definitions so that we can later
apply this theorem.

A program shape $s$ is the pairing of a mapping from component names
to component interfaces and a set that indicates uncompromised
components.
In the rest of the paper, we implicitly restrict our attention
to \emph{well-formed} shapes.
A shape is well-formed when (1) all component interfaces in the shape
only import procedures from components that are part of the shape, and
(2) these procedures are exported according to the shape.

High-level partial programs $P$ and contexts $A$ are defined as
mappings from component names to component definitions.
A high-level partial program $P$ has shape ${\bullet}s$ when it
defines exactly the components that are marked as \emph{uncompromised}
in $s$, with definitions that satisfy the corresponding interfaces,
and when it moreover satisfies the simple well-formedness condition
that all the local buffers it uses are defined.
A high-level context $A$ has shape ${\circ}s$ under the same
conditions, adapted for \emph{compromised} components instead of
uncompromised ones.

A low-level partial program $p$ or context $a$ is formed by pairing a
partitioned memory with a mapping from procedure identifiers to entry
points.
This choice is isomorphic to having sets of named compartment
memories with entry point annotations.
A low-level partial program $p$ has shape ${\bullet}s$ when the
partitioned memory has partitions under exactly the component names
that are marked as \emph{uncompromised} in $s$, and the entry point
mapping provides addresses for exactly the procedures that are
exported by these components according to $s$.
A low-level context $a$ has shape ${\circ}s$ under the same
conditions, adapted for \emph{compromised} components instead of
uncompromised ones.

We say that a high-level partial program $P \psh{} s$ is fully defined
with respect to contexts of shape ${\circ}s$ when it cannot be blamed
for undefined behavior when interacting with such contexts:
for every $A \ash{} s$, either reducing $A[P]$ never reaches an
undefined behavior configuration, or else the current component in
this undefined behavior configuration belongs to $A$.
Similarly, a high-level context $A \psh{} s$ is fully defined
with respect to programs of shape ${\circ}s$ when it cannot be blamed
for undefined behavior when interacting with such programs.

Because we perform a point-wise compilation of high-level programs,
separate compilation (\autoref{defn:separate-compilation}) trivially
holds for our compiler.
Combining it with whole-program compiler correctness
(\autoref{assumption:compiler-correctness}) immediately leads to the
following corollary:

\begin{cor}[Separate compilation correctness]\Coqed
\label{cor:separate-compiler-correctness}
\[
\begin{array}{l}
\forall s, A \ash{}s, P \psh{}s. \\
~~ P\text{ fully defined wrt. contexts of shape }{\circ}{s} \Rightarrow \\
~~ A\text{ fully defined wrt. programs of shape
}{\bullet}{s} \Rightarrow \\
\text{~~~~(1) }\TERM{A[P]} \iff \TERM{\comp{A}[\comp{P}]}~\wedge \\
\text{~~~~(2) }\DIV{A[P]} \iff \DIV{\comp{A}[\comp{P}]}
\end{array}
\]
\end{cor}

\SUBSECTION{Proof technique for structured full abstraction}
\label{sec:proof-outline}

Trace semantics were initially proposed by Jeffrey and
Rathke~\cite{JeffreyR05b,JeffreyR05} to define fully abstract models
for high-level languages.
Patrignani~\ETAL later showed
how to use trace semantics~\cite{PatrignaniC15} to prove full
abstraction for a compiler targeting machine
code~\cite{PatrignaniASJCP15}.

This proof technique is well suited for deterministic target languages
such as machine code and proceeds in two steps.
First, we devise a trace semantics for low-level partial programs and
contexts and relate it to the target machine's operational semantics (\EG by
proving it fully abstract~\cite{PatrignaniC15}).
This trace semantics will provide a set of traces for every partial
program, describing all the execution paths that this program
can reach by interacting with an arbitrary context.
Second, we use the trace semantics to characterize the interaction
between an arbitrary low-level context and, separately, two compiled
programs that this context distinguishes, resulting in two traces with a
common prefix followed by different actions.
We can then use these traces to construct a high-level attacker,
proving that this attacker distinguishes between the two source programs.

As our proof demonstrates, proving the trace semantics fully abstract
is not a mandatory first step in the technique.
Instead, we relate our trace semantics to the operational one
using two weaker trace composition and decomposition conditions
(\autoref{lemma:trace-decomp}
and \autoref{lemma:trace-comp}),
adapted from the key lemmas that Jeffrey and Rathke used to
prove their trace semantics fully
abstract~\cite{JeffreyR05b,JeffreyR05}.
This reduces proof effort, since proving a trace semantics fully
abstract typically requires proving a third lemma with a
trace-mapping argument of its
own~\cite{PatrignaniC15,JeffreyR05b,JeffreyR05}.

%

Adapting the technique to undefined behavior is straightforward, 
essentially amounting to proving standard full abstraction for the safe
subset of the language.
Then one simply proves that the context produced by the
mapping is fully defined, thus safe.
Adapting to a closed world, however, takes more work.

The trace semantics that have been used previously to prove full abstraction
characterize the interaction between a partial program and arbitrary
contexts.
The context's shape is typically constructed as reduction goes,
based on the steps that the context takes in the generated trace.
For instance, if the trace said that the context performs a call, then
the target procedure would be appended to the context's interface so
that this call becomes possible.
For structured full abstraction, we want a finer-grained trace
semantics that enables reasoning about the interaction with
contexts of a specific shape.
We achieve this by making the shape a parameter to the reduction
relation underlying our trace semantics.
To make sure that traces are compatible with this shape, we also
keep track of the current compartment during reduction.
This allows us to generate only context steps that adhere to the
current compartment's interface, and hence to the context's shape.
In particular, the context will only be able to call program
procedures for which
(1) there is a context compartment whose interface explicitly imports
the target procedure, thus granting the privilege to call that
procedure,
and (2) this other context compartment is reachable from the current
compartment via a chain of cross-compartment calls or returns within
the context.

Moving to a closed world also makes the trace mapping argument harder.
The one from Patrignani~\ETAL~\cite{PatrignaniASJCP15},
for instance, relies on changes in the context's shape,
\EG adding a helper component to the context that is not present in
the low level.
This is no longer possible for structured full abstraction,
where the context shape is fixed.

\SUBSECTION{Trace semantics for the low-level language}
\label{sec:trace-semantics}

We define a trace semantics in which traces are finite words over an
alphabet $E\alpha${}
of external actions, alternating between program external
actions ``$\gamma!$'' and context external actions ``$\gamma?$''.
We treat external actions as moves in a two-player game, viewing
the context and the partial program as the players.
The trace semantics is parameterized by a shape $s$, which the two
players have.
External actions either transfer control to the other player or end
the game.
\[
  \begin{array}{r@{}l@{~~~~~~~}r@{}l}
  E\alpha\  \mathrel{::=}\; \ \;&
    \gamma!
    \;|\;
    \gamma?
  &
  \gamma \ \mathrel{::=}\; \ \;&
    \con{Call}_{\ii{reg}}~C~P
    \;|\;
    \con{Return}_{\ii{reg}}
    \;|\;
    \checkmark
  \end{array}
\]
%
%
%
Traces ($E\alpha^*$) track the external actions ($\gamma$)
performed by the context and the program.
The first kind of external action is cross-boundary communication,
which corresponds to the use of instrumented call instructions
$\con{Call}~C~P$ and $\con{Return}$ when they transfer control to
a compartment that belongs to the opponent.
For these external actions, traces keep track of the instruction used
together with $\ii{reg}$, the values held by all registers when the
instruction is issued.
The second kind of external action is program termination,
which we denote with a tick $\checkmark$ and which the opponent cannot
answer ($\checkmark$ ends the game).
It corresponds to the use of an instruction that makes execution
stuck, such as $\con{Halt}$.

At any point where it has control, a player can take internal actions
(any instruction that
neither terminates execution nor transfers control to the opponent);  
these are not reflected in the trace.
In particular, cross-compartment communication is considered an
{\em internal} action when it transfers control to a compartment that
belongs to the current player.
Besides halting, a player can also end the game by triggering an infinite
sequence of internal actions, making execution diverge.
In the trace, this will correspond to not making any move:
the trace observed thus far will be a maximal trace for the
interaction between the program and context involved, \IE any
extension of this trace will not be shared by both the program and the
context.

Intuitively, a program $p \psh{} s$ has trace $t$ if it answers with the
{program} actions described in $t$ when facing a context $a \ash{}s$ that
produces the {context} actions described in $t$.
Similarly, a program $a \ash{}s$ has trace $t$ if it answers with the
{context} actions described in $t$ when facing a program $p \psh{}s$
that produces the {program} actions described in $t$.
We define $\ptr{s}{p}$ to be the set of traces of a partial program
$p$ with respect to contexts of shape ${\circ}s$, and $\atr{s}{a}$
to be the set of traces of a context $a$ with respect to programs of
shape ${\bullet}s$.

The player that starts the game is the one that owns the main
component according to $s$.
For each player, the trace semantics is deterministic with respect to
its own actions and nondeterministic with respect to the opponent's
actions.
All possible actions an actual opponent could take have a
corresponding nondeterministic choice, which is formalized by a
property we call trace extensibility.

\begin{lemma}[Trace extensibility]
\label{lem:trace-ext}
\[
\begin{array}{l}
\forall t,\; s,\; p \psh{}s,\; a \ash{}s. \\
  ~~ ( t \in \ptr{s}{p} \wedge t.\gamma? \in \atr{s}{a} \Rightarrow
       t.\gamma? \in \ptr{s}{p} ) ~ \wedge \\
  ~~ ( t \in \atr{s}{a} \wedge t.\gamma! \in \ptr{s}{p} \Rightarrow
       t.\gamma! \in \atr{s}{a} )
\end{array}
\]
\end{lemma}

Nondeterminism disappears once we choose a particular opponent for
a player, as the two key lemmas below illustrate.

\begin{lemma}[Trace decomposition]
\label{lemma:trace-decomp}
\[
\begin{array}{l}
  \forall s,\; p \psh{}s,\; a \ash{}s.~~
  \TERM{a[p]} \Rightarrow \\
  ~~\exists t.~~ \TTERM{t} \wedge t \in \ptr{s}{p} \cap \atr{s}{a}
\end{array}
\]
\end{lemma}

Trace decomposition is stated for terminating programs.
It extracts the interaction\ch{an interaction? unique or not?}
between a program $p$ and a context $a$
with dual shapes by looking at how $a[p]$ reduces, synthesizing that
interaction into a trace $t$.
Because execution terminates, this trace ends with a
termination marker.

\begin{lemma}[Trace composition]
\label{lemma:trace-comp}
\[
\begin{array}{l}
  \forall t,\; s,\; p \psh{}s,~ a \ash{}s.~~
  t \in \ptr{s}{p} \cap \atr{s}{a} \Rightarrow \\
  ~~ ( \forall E\alpha.~(t.E\alpha) \not\in \ptr{s}{p} \cap \atr{s}{a} ) \Rightarrow \\
  ~~~~~(\TERM{a[p]} \iff \TTERM{t})
\end{array}
\]
\bcp{Is $E\alpha$ a variable ranging over $E\alpha$??}
\end{lemma}

Trace composition is the opposite of trace decomposition, 
reconstructing a sequence of reductions based on synthesized interaction
information.
It considers a program and a context with dual shapes, that share a
common trace $t$.
The condition on the second line states that the game has ended:
trace $t$ cannot be extended by any action $E\alpha$ such that the
two players share trace ``$t.E\alpha$''.
Under these assumptions, trace composition tells us that one of the
following holds:
either (1) the trace ends with a termination marker $\checkmark$ and
putting $p$ in context $a$ will produce a terminating
program,
or (2) the trace does not end in $\checkmark$
and putting $p$ in context $a$ will produce a diverging program.
Intuitively, if the game has ended but there is no termination
marker, it must be because one of the players went into an
infinite sequence of internal actions and will neither give control
back nor terminate.

While the statement of these lemmas is quite close to that used in
an open world setting~\cite{JeffreyR05b,JeffreyR05},
the trace semantics itself has to be adapted in order to prove them
in the presence of our closed world assumption.
To this end, we incorporate \emph{internal} actions within the trace
semantics, thus adding more options to the nondeterministic
choice of the next context action, which allows us to track at any
point the currently executing compartment.
When in control, a player can only perform communicating actions
allowed by the interface of the current compartment.
This restricts external actions as required, while also making it possible
to internally switch the current compartment through allowed internal
actions.
Using our semantics, we thus end up with finer-grained traces that
include internal communication, which can be directly mapped to
high-level attackers (\autoref{assumption:definability}).
The traces we use otherwise are obtained by erasing internal actions
from the finer-grained traces.
%

\SUBSECTION{Proof of SCC}
\label{sec:md-proof}

%
We prove our instance of structured full abstraction,
which implies SCC by \autoref{thm:sfa-to-sc} since we have
isomorphic definitions to the ones in \autoref{sec:fa-not-enough}.

\begin{thm}[Structured full abstraction]\label{thm:sfa-instance}\Coqed\\
Our compiler satisfies structured full abstraction.
\end{thm}


Recall that the basic idea behind the proof technique is to extract
two traces that characterize the interaction between a low-level context
and two compiled fully defined high-level programs, and then to map these
two traces to a fully defined high-level context.
The high-level context should reproduce the context actions described
in the traces when facing the same programs as the low-level context.

Unfortunately, a compiled fully defined context cannot reproduce
any arbitrary low-level trace, because the values transmitted
in registers are part of external communication actions in low-level
traces:
As enforced by the compiler, these contexts always clear all
registers but the one used for communication before giving control to
the program.
They can thus only produce traces in which registers are
cleared in all context actions, which we call \emph{canonical} traces.
We denote by $\zeta(\gamma)$ the operation that rewrites action
$\gamma$ so that all registers but one are clear.
A canonical trace $\zeta_{\circ}(t)$ can be obtained from an arbitrary
trace $t$ by replacing all context actions ``$\gamma?$'' by
``$\zeta(\gamma)?$''.
We call this operation trace canonicalization.

As we will see, being able to reproduce arbitrary canonical
traces gives enough distinguishing power to the high-level context.
The reason is that, because they can't trust other compartments,
compiled fully defined components never read values transmitted in
registers with the exception of the one used for communication.
As a consequence, these components cannot distinguish context external actions
based on the content of these unread registers, which are exactly the
ones a compiled fully defined context cleans.
Fully defined programs thus perform the exact same actions when facing
a trace $t$ or its canonicalization $\zeta_{\circ}(t)$,
as formalized by \autoref{lem:canonicalization}.
This means that having the high-level attacker reproduce canonical
traces instead of the original traces of the low-level context will be
enough to lead compiled programs into reproducing the actions they took when
facing the low-level context.

\begin{lemma}[Canonicalization]
\label{lem:canonicalization}
\[
\begin{array}{l}
\forall t,\; s,\; P \psh{}s. \\
~~
P\text{ fully defined wrt. contexts of shape }{\circ}s \Rightarrow \\
~~~~  t \in \ptr{s}{\comp{P}} \iff \zeta_{\circ}(t) \in \ptr{s}{\comp{P}}
\end{array}
\]
\end{lemma}

The definability assumption{}
below gives a characterization of our 
mapping from a canonical trace $t$ and an action $\gamma_1$
to a compiled fully defined context $\comp{A}$
that reproduces the context actions in $t$ and,
depending on the next action $\gamma$ the
program takes, ends the game with either termination (if
$\zeta(\gamma) = \zeta(\gamma_1)$) or divergence
(if $\zeta(\gamma) \not= \zeta(\gamma_1$)).
The context $\comp{A}$ will thus distinguish a program $p$ producing
trace ``$t.\gamma_1!$'' from any program producing ``$t.\gamma!$'' with
$\zeta(\gamma) \not= \zeta(\gamma_1)$.

\begin{assumption}[Definability]
\label{assumption:definability}
\[
\begin{array}{l}
  \forall t,\, \gamma_1,\, s.~~
  t = \zeta_{\circ}(t) \wedge
  (\exists p \psh{}s.~ (t.\gamma_1!) \in \ptr{s}{p}) \Rightarrow \\
  ~~\exists A \ash{}s.~ A\text{ fully defined wrt. programs of shape }{\bullet}s ~\wedge \\
  \text{~~~~(1) }t \in \atr{s}{\comp{A}} ~\wedge \\
  \text{~~~~(2) }( \gamma_1 \neq \checkmark \Rightarrow
    (t.\gamma_1!.\checkmark?) \in \atr{s}{\comp{A}} ) ~\wedge \\
  \text{~~~~(3) }\forall \gamma.\;\text{if }
  \zeta(\gamma) \not= \zeta(\gamma_1) \text{ then } \forall \gamma'.~
  (t.\gamma!.\gamma'?) \not\in \atr{s}{\comp{A}}
\end{array}
\]
\end{assumption}

The definability assumption gives us a fully defined
context that follows trace $t$ (part 1) and that, if given
control afterwards via action ``$\gamma!$'' such that
$\gamma \not= \checkmark$, acts as follows:
if $\gamma = \gamma_1$ the context terminates (2)
and if the context can distinguish $\gamma$ from $\gamma_1$, it
will make execution diverge by not issuing any action $\gamma'$ (3).
Since it is a compiled fully defined context, $\comp{A}$ can only access
values transmitted using register $r_\ii{com}$, the register that
holds the call argument or return value.
So $\comp{A}$ can only distinguish between $\gamma$ and $\gamma_1$
when they differ in $r_\ii{com}$, which is captured formally by
the $\zeta(\gamma) \neq \zeta(\gamma_1)$ condition.
\ch{What about completely different actions? Say a call and a return,
  or two calls to completely different places?}

Proving this assumption (even on paper) would be quite tedious,
so we settled for testing its correctness using
QuickCheck~\cite{ClaessenH00}.
We built an algorithm (in OCaml) that constructs $A$ out of $t$.
\ch{Confused by the ``More precisely'' below and how
  it relates to the formal assumption above. Are you testing
  the assumption above or something slightly different?}
\ch{It was actually more than just slightly different}
More precisely, the algorithm inputs a trace with internal
actions (the finer-grained trace that erases to $t$)
and builds a context $A$ that reproduces context internal and
external actions as prescribed by that trace.
Execution will afterwards resume at a different point in $A$ depending
on the next action taken by the program.
At each such point, $A$ will either terminate execution or make it
diverge depending on whether the program action is distinguishable
from action $\gamma_1$.
Because the trace taken as input already includes internal actions, we do
not have to reconstruct them, hence our algorithm is not more
difficult to devise than in an open-world setting~\cite{PatrignaniASJCP15}.
In the following, we assume that the algorithm is correct,
\IE that \autoref{assumption:definability} holds.
We can now turn to the main theorem.

\begin{proof}[Detailed proof of structured full abstraction]
Consider a low-level attacker $a \ash{}s$ distinguishing
two fully defined partial programs  $P,Q \psh{}s$ after compilation.
Suppose without loss of generality that $a[\comp{P}]$ terminates and $a[\comp{Q}]$
diverges.
We build a high-level attacker $A \ash{}s$ that is fully defined
with respect to programs of shape $\bullet s$ and can distinguish between
$P$ and $Q$.

We can first apply trace decomposition (\autoref{lemma:trace-decomp})
to $a$ and $\comp{P}$ to get a trace $t_i \in \ptr{s}{\comp{P}}$ that ends
with $\checkmark$, such that $t_i \in \atr{s}{a}$.
Call $t_p$ the longest prefix of $t_i$ such that
$t_p \in \ptr{s}{\comp{Q}}$.
Because trace sets are prefix-closed by construction, we know that
$t_p \in \ptr{s}{\comp{P}} \cap \atr{s}{a}$.

Moreover, $t_p$ is necessarily a \emph{strict} prefix of $t_i$:
otherwise, we could apply trace composition
(\autoref{lemma:trace-comp})
and get that $a[\comp{Q}]$ terminates, a contradiction.
So there exists an
external action $E\alpha$ such that trace ``$t_p.E\alpha$'' is a
prefix of $t_i$.
Now $E\alpha$ cannot be a context action,
or else trace extensibility (\autoref{lem:trace-ext}) would imply that
``$t_p.E\alpha$'' is a trace of $\ptr{s}{\comp{Q}}$, which is incompatible
with $t_p$ being the \emph{longest} prefix of $t_i$ in
$\ptr{s}{\comp{Q}}$.
Therefore, $E\alpha$ is a program action, i.e., there
exists $\gamma_1$ such that ``$E\alpha = \gamma_1!$''.
Intuitively, $\comp{P}$ and $\comp{Q}$ take the same external actions
until the end of $t_p$, where $\comp{P}$ takes external action
``$\gamma_1!$'' and $\comp{Q}$ does not (it takes either a different
action $\gamma \neq \gamma_1$ or no external action at all).

Now, let $t_c$ be the canonicalization of trace $t_p$, i.e.,
$t_c = \zeta_{\circ}(t_p)$.
By canonicalization (\autoref{lem:canonicalization}),
``$t_c.\gamma_1!$'' $= \zeta_{\circ}(t_p.\gamma_1!)$ is a trace of $\comp{P}$.
We can thus use apply definability (\autoref{assumption:definability})
to trace $t_c$ and action $\gamma_1$, using $\comp{P} \psh{}s$ as a
witness having trace ``$t_c.\gamma_1!$''.
This yields a {fully defined} context $A \ash{}s$ such that:
\[
\begin{array}{l}
  \text{(1) }t_c \in \atr{s}{\comp{A}}, \\
  \text{(2) }\gamma_1 \neq \checkmark \Rightarrow
  (t_c.\gamma_1!.\checkmark?) \in \atr{s}{\comp{A}}, \\
  \text{(3) }\forall \gamma,~ \gamma'.~~(t_c.\gamma!.\gamma'?) \in \atr{s}{\comp{A}} \Rightarrow

        \zeta(\gamma) = \zeta(\gamma_1).
\end{array}
\]
We now show that these conditions imply that
$\comp{A}\![\comp{P}]$
terminates while $\comp{A}[\comp{Q}]$ diverges.

First, we look at $\comp{P}$.
Consider the case where
$\gamma_1 = \checkmark$.
In this case, by applying trace extensibility to $\comp{A}$ in (1), we
get that
``$t_c.\checkmark!$'' is a trace of $\comp{A}$, so trace
composition allows us to conclude that $\comp{A}[\comp{P}]$
terminates.
Now if $\gamma_1 \neq \checkmark$ then this action gives back control
to the context, which, given (2), will perform action ``$\checkmark?$''.
Applying trace extensibility to $\comp{P}$, $\comp{P}$ has trace
``$t_c.\gamma_1!.\checkmark?$'', so we can apply trace composition and
deduce that $\comp{A}[\comp{P}]$ terminates in this case as well.

Now, regarding $\comp{Q}$, we first obtain the following by applying
canonicalization to $t_p$, ``$t_p.\checkmark!$'', and
``$t_p.\gamma_1!$'':
\[
\begin{array}{l}
  \text{(a) }t_c = \zeta_{\circ}(t_p) \in \ptr{s}{\comp{Q}}, \\
  \text{(b) }(t_c.\checkmark!) = \zeta_{\circ}(t_p.\checkmark!) \in \ptr{s}{\comp{Q}} \Rightarrow
  (t_p.\checkmark!) \in \ptr{s}{\comp{Q}}, \\
  \text{(c) }(t_c.\gamma_1!) = \zeta_{\circ}(t_p.\gamma_1!) \in \ptr{s}{\comp{Q}} \Rightarrow
  (t_p.\gamma_1!) \in \ptr{s}{\comp{Q}}.
\end{array}
\]

After following trace $t_c$, which $\comp{Q}$ has from (a), $\comp{Q}$
cannot perform a terminating action:
otherwise using (b) and trace extensibility for $a$ and $t_p$, we
could apply trace composition to trace ``$t_p.\checkmark$'' and get that
$a[\comp{Q}]$ terminates, which is a contradiction.
$\comp{Q}$ cannot perform action $\gamma_1$ either, since (c) would
then violate the fact that $t_p$ is the longest prefix of $t_i$ in
$\ptr{s}{\comp{Q}}$.
So $\comp{Q}$ only has two options left.
The first is to perform no external action by
going into an infinite sequence of internal transitions.
In this case, using (1), we can apply trace composition to get that
$\comp{A}[\comp{Q}]$ diverges.
The second option is to give control back to the context using
an external action $\gamma$ so that
$\checkmark \neq \gamma \neq \gamma_1$.
Because fully defined compiled programs clean registers, they only
yield canonical actions, i.e.
$\gamma = \zeta(\gamma) \wedge \gamma_1 = \zeta(\gamma_1)$.
Combined with (3), this entails that if $\comp{A}$ produced an
action $\gamma'$, we would have $\gamma = \gamma_1$, which is
false.
Hence, $\comp{A}$ doesn't produce any action: it goes into an infinite
sequence of local transitions.
We can again apply trace composition to get that
$\comp{A}[\comp{Q}]$ diverges.

We finally apply separate compiler correctness
(\autoref{cor:separate-compiler-correctness}) to conclude the proof.
\end{proof}

\section{Related Work}
\label{sec:related}

\SUBSECTION{Fully abstract compilation}

Fully abstract compilation was introduced in the seminal work of
Mart\'in Abadi~\cite{abadi_protection98}
%
%
and later investigated by the academic community.
(Much before \iffull\else this\fi,
 the concept of full abstraction
 was coined by Milner~\cite{Milner75}.)
\iffull\else For instance, \fi
Ahmed~\ETAL\cite{AhmedB08,AhmedB11, Ahmed15, NewBA16} proved the
full abstraction of type-preserving compiler passes for functional
languages and devised proof techniques for {\em typed} target languages.
Abadi and Plotkin~\cite{abadi_aslr12} and
Jagadeesan~\ETAL\cite{JagadeesanPRR11} expressed the protection
provided by a mitigation technique called address space layout
randomization as a probabilistic variant of full abstraction.
Fournet~\ETAL\cite{FournetSCDSL13} devised a fully abstract compiler
from a subset of ML to JavaScript.

\ch{Cedric Fournet asked Marco at CSF about the ``modularity''
  property, and whether it's equivalent to (basically) our separate
  compilation, but it seemed that he already had work on ``separate
  compilation'' ... maybe in the process calculus setting. Can we dig
  that up and give him proper credit if needed?}

Patrignani~\ETAL\cite{PatrignaniASJCP15, LarmuseauPC15}
were recently the first to study fully abstract compilation to machine code,
starting from single modules written in simple, idealized
object-oriented and functional languages and targeting hardware
architectures featuring a new coarse-grained isolation mechanism.
%
Until recently, they studied fully abstract compilers that
by design violate our separate compilation property, so they cannot be
applied to our compartmentalizing compilation setting.

In recent parallel work, Patrignani~\ETAL\cite{PatrignaniDP16}
proposed an extension of their compilation scheme to protecting
multiple components from each other.
The attacker model they consider is different, especially since their
source language does not have undefined behavior.
Still, if their source language were extended with unsafe features,
our SCC property could possibly hold for their compiler.

Patrignani, Devriese~\ETAL\cite{PatrignaniC15, DevriesePP16} also
proposed proof techniques for full abstraction that work for untyped
target languages, and more recently proved full abstraction by
approximate back-translation for a compiler between the simply typed
and the untyped lambda calculus and fully formalized this proof in
Coq~\cite{DevriesePPK17}. As opposed to our Coq proof for the instance
from \autoref{sec:instance}, which is done under well-specified
assumptions but not complete, the proof of Devriese~\ETAL is assumption free.

\ch{In this work we tease out the unsafe source language problem and
  beat it to death. We could extend the compartmentalizing compilation
  story to also encompass secure interoperability with mutually
  distrustful low-level components, as we already did informally in
  \cite[\S 2.2 and \S 2.3]{JuglaretHAPST15}, and as they do in
  \cite{PatrignaniDP16}. The key point there is that low-level
  components can also get protection, but only if they are
  observationally equivalent to compiled components, which is a strong
  restriction. The most realistic scenario we could find where this
  could be useful was hand-optimizing assembly produced by our
  compiler \cite[\S 2.3]{JuglaretHAPST15}, but \cite{PatrignaniDP16}
  came up with more scenarios where ``modularity'' is supposed to
  help, although these scenarios often don't have much to do with
  security (\EG separate compilation is a good feature to have, even
  if things are compiled one might not have access to the source,
  independent attestation for PMA). If one doesn't find these use
  cases for mutual distrust convincing with low-level components (\EG
  because increasing overhead without clear security benefits is a
  hard sell), one might be tempted to just say: use full abstraction
  there and merge all low-level components into one. Anyway, could all
  this be worth discussing in this paper?}

\ch{A different reason to include low-level components is horizontal
  composition for secure compilers. That's a challenging topic;
  see comment in by ERC b2 after discussion with Yannis.}

\SUBSECTION{Formal reasoning about compartmentalized code}


SCC is orthogonal to formal techniques for reasoning about
compartmentalized software: SCC allows {\em transferring} security
guarantees for compartmentalized code written in a source language
to machine code via compartmentalizing compilation, but SCC itself
does not provide effective reasoning principles to obtain those
security guarantees in the first place.
The literature contains interesting work on formally characterizing
the security benefits of compartmentalization.
Promising approaches include Jia~\ETAL's work on System
M~\cite{JiaS0D15}, and Devriese~\ETAL's work on logical relations for
a core calculus based on JavaScript~\cite{DevriesePB16}, both of which
allow bounding the behavior of a program fragment based on the
interface or capabilities it has access to.
One significant challenge we attack in this paper is languages with
undefined behaviors, while in these other works illegal actions such
as accessing a buffer out of bounds must be detected and make the program
halt.

\SUBSECTION{Verifying correct low-level compartmentalization}

Recent work focused on formally verifying the correctness of
low-level compartmentalization
mechanisms based on software fault isolation~\cite{ZhaoLSR11,
  MorrisettTTTG12, KrollSA14} or tagged hardware~\cite{micropolicies2015}.
That work, however, only considers the correctness of the low-level
compartmentalization mechanism, not the compiler and not high-level
security properties and reasoning principles for code written in
a programming language with components.
Communication between low-level compartments is generally done by jumping to a
specified set of entry points, while the model we consider in
\autoref{sec:instance} is more structured and enforces correct calls
and returns.

Finally, seL4 is a verified operating system
microkernel~\cite{Klein09sel4:formal}, that uses a capability system
to separate user level threads and for which correct access
control~\cite{SewellWGMAK11} and noninterference
properties~\cite{seL4:Oakland2013} were proved formally.

\section{Conclusion and Future Work}
\label{sec:conclusion}

We have introduced a new secure compartmentalizing compilation property,
related it to the established notion of full abstraction,
and applied our property in a carefully simplified setting: a small
imperative language with procedures compiling to a compartmentalized
abstract machine.
This lays the formal foundations for studying the secure compilation
of mutually distrustful components written in unsafe languages.

In the future we plan to build on this groundwork to study more realistic
source and target languages, compilers, and enforcement mechanisms.
In the long run, we would like to apply this to the C language by
devising a secure compartmentalizing variant of CompCert that targets a
tag-based reference monitor~\cite{micropolicies2015} running on a real
RISC processor~\cite{Dover16}.
We have in fact started working towards this long term
goal~\cite{JuglaretHAPST15}, but this will take time to achieve.
Beyond tagged hardware, we would also like to implement the abstract
compartmentalization machine from \autoref{sec:instance} in terms of
various enforcement other mechanisms, including: process-level
sandboxing~\cite{Kilpatrick03, GudkaWACDLMNR15, wedge_nsdi2008,
  ProvosFH03}, software fault isolation (SFI)~\cite{YeeSDCMOONF10},
capability machines~\cite{cheri_oakland2015}, and multi-PMA
systems~\cite{PatrignaniDP16}.
As we target lower-level machines, new problems will appear: for
instance we need to deal with the fact that memory is finite and
resource exhaustion errors cannot be hidden from the attacker, which
will require slightly weakening the security property.
%
%
Finally, we would like to study more interesting compartmentalization
models including dynamic component creation and nested components, and
the way these extensions influence the security property.




\ifanon\else
\paragraph*{Acknowledgments}
We thank
Andr\'e DeHon,
Deepak Garg, and
Andrew Tolmach
for helpful discussions and thoughtful feedback on earlier drafts.
Yannis Juglaret is supported by a PhD grant from the
French Department of Defense (DGA) and Inria.
Arthur Azevedo de Amorim and Benjamin C. Pierce are
supported by
\href{http://www.nsf.gov/awardsearch/showAward?AWD_ID=1513854}{NSF award 1513854,
{\em Micro-Policies: A Framework for Tag-Based Security Monitors}}.

\fi

\bibliographystyle{plainurl}
\bibliography{mp,safe}

\end{document}